\documentclass[asa,12pt,psfig,twoside]{article}
\usepackage{amsmath}
\usepackage{amsthm}
\usepackage{mathptmx}
\usepackage{graphicx}
\usepackage{float}
\usepackage[titletoc,title]{appendix}
\usepackage[labelfont=bf]{caption} 
\usepackage[a4paper,margin=1in]{geometry}
\usepackage[hang,flushmargin]{footmisc}
\usepackage[%
colorlinks=true,%
urlcolor=blue%
]{hyperref}
\usepackage{apacite} 
\usepackage{sectsty}

\sectionfont{\fontsize{12}{15}\selectfont}
\subsectionfont{\fontsize{12}{15}\selectfont}
\subsubsectionfont{\fontsize{12}{15}\selectfont}

\usepackage{titlesec}
\titlelabel{\thetitle.\quad}
\titlespacing\section{-5pt}{12pt plus 4pt minus 2pt}{0pt plus 2pt minus 2pt}
\titlespacing\subsection{-5pt}{12pt plus 4pt minus 2pt}{0pt plus 2pt minus 2pt}
\titlespacing\subsubsection{-5pt}{12pt plus 4pt minus 2pt}{0pt plus 2pt minus 2pt}

\renewcommand{\theequation}{\arabic{equation}}
\def\mybibliography#1{{\noindent \Large \bf References}\list
 {}{\setlength{\leftmargin}{0.4in}\setlength{\labelsep}{0pt}
\itemindent=-\leftmargin}
 \def\newblock{\hskip .02em plus .20em minus -.07em}
 \sloppy\clubpenalty4000\widowpenalty4000
 \sfcode`\.=1000\relax}
\newbox\TempBox \newbox\TempBoxA
\def\uw#1{%
  \ifmmode\setbox\TempBox=\hbox{$#1$}\else\setbox\TempBox=\hbox{#1}\fi%
  \setbox\TempBoxA=\hbox to \wd\TempBox{\hss\char'176\hss}%
  \rlap{\copy\TempBox}\smash{\lower9pt\hbox{\copy\TempBoxA}}%
}
\newbox\TempBox \newbox\TempBoxA
\def\uwd#1{%
  \ifmmode\setbox\TempBox=\hbox{$#1$}\else\setbox\TempBox=\hbox{#1}\fi%
  \setbox\TempBoxA=\hbox to \wd\TempBox{\hss\char'176\hss}%
  \rlap{\copy\TempBox}\smash{\lower10pt\hbox{\copy\TempBoxA}}%
}
\def\mathunderaccent#1{\let\theaccent#1\mathpalette\putaccentunder}
\def\putaccentunder#1#2{\oalign{$#1#2$\crcr\hidewidth
\vbox to.2ex{\hbox{$#1\theaccent{}$}\vss}\hidewidth}}

\usepackage{setspace}
\usepackage{indentfirst} 

\usepackage{fancyhdr}
\fancypagestyle{firstpage}{%
	\fancyhf{}%
	\setlength{\headsep}{1pt}
	\lhead{\textit{Statistics and Application} {\bf{\{ISSN 2452-7395(online)\}}}\\
	Volume .., Nos. …, … (New Series), pp ...‒…}
	
}

\usepackage{lipsum}

\newcommand\blfootnote[1]{%
	\begingroup
	\renewcommand\thefootnote{}\footnote{#1}%
	\addtocounter{footnote}{-1}%
	\endgroup
}

\newcommand{\ytil}{\mathunderaccent\tilde{y}}
\newcommand{\mutil}{\mathunderaccent\tilde{\mu}}

\newcommand{\wtil}{\mathunderaccent\tilde{w}}

\newcommand{\betatil}{\mathunderaccent\tilde{\beta}}
\newcommand{\xtil}{\mathunderaccent\tilde{x}}
\newcommand{\tautil}{\mathunderaccent\tilde{\tau}}
\newcommand{\gtil}{\mathunderaccent\tilde{g}}
\newcommand{\gmtil}{\mathunderaccent\tilde{\gamma}}
\newcommand{\ometil}{\mathunderaccent\tilde{\omega}}
\newcommand{\jtil}{\mathunderaccent\tilde{j}}
\newcommand{\titlesize}{\fontsize{16pt}{20pt}\selectfont}
\newcommand{\namesize}{\fontsize{12pt}{20pt}\selectfont}
\renewenvironment{abstract}
{\begin{quote}
		\noindent \rule{\linewidth}{.5pt}\par{\bfseries \abstractname}}
	{\medskip\noindent \rule{\linewidth}{.5pt}
	\end{quote}

} 

\begin{document}

\newtheorem{theorem}{Theorem}[section]

\newtheorem{proposition}{Proposition}[section]
\newtheorem{corollary}{Corollary}[section]
\newtheorem{lemma}{Lemma}[section]

\hypersetup{linkcolor=black}

\vspace*{.05in}

\begin{center}

{\titlesize \textbf{Bayesian Logistic Regression for \\ \vspace{6pt}
Small Areas with Numerous Households}}

\vspace*{12pt}
{\namesize \bf Balgobin Nandram\blfootnote{Corresponding Author: Balgobin Nandram\\E-mail:\href{mailto:balnan@wpi.edu}{balnan@wpi.edu} }, Lu Chen, Shuting Fu and Binod Manandhar}\\
\textit{Department of Mathematical Sciences, Worcester Polytechnic Institute,\\
	 Worcester, MA, USA}\\
\vspace*{12pt}
Received ….; Revised ….; Accepted …..
\end{center}

\begin{abstract}
		
\setlength{\parindent}{3em}
  We analyze binary data, available for a relatively large number (big data) of families (or households), which are within small areas, from a population-based survey. Inference is required for the finite population proportion of individuals with a specific character for each area. To accommodate the binary data and important features of all sampled individuals, we use a hierarchical Bayesian logistic regression model with each family (not area) having its own random effect. This modeling helps to correct for overshrinkage so common in small area estimation. Because there are numerous families, the computational time on the joint posterior density using standard Markov chain Monte Carlo (MCMC) methods is prohibitive. Therefore, the joint posterior density of the hyper-parameters is approximated using an integrated nested normal approximation (INNA) via the multiplication rule. This approach provides a sampling-based method that permits fast computation, thereby avoiding very time-consuming MCMC methods. Then, the random effects are obtained from the exact conditional posterior density using parallel computing. The unknown nonsample features and household sizes are obtained using a nested Bayesian bootstrap that can be done using parallel computing as well. For relatively small data sets  (e.g., 5000 families), we  compare our method with a MCMC method to show that our approach is reasonable. We discuss an example on health severity using the Nepal Living Standards Survey (NLSS). 
  
	\medskip
	
	\noindent {\textit{Keywords}}: Bayesian bootstrap, Big data, Integrated nested normal approximation, Overshrinkage, 
	Sampling based method, Survey weights.
	
\end{abstract}

\setcounter{page}{1}
\thispagestyle{firstpage}

\newpage

\section{Introduction}
\label{S:1}

In the second Nepal Living standards Survey (NLSS II), there are data from households. One question of interest is health status (good versus poor health), a binary variable, and there are several covariates that can explain the binary outcomes. Our interest is to provide smoothed
estimates of the household proportions of members in good health for both sampled and nonsampled households. This is a general question.
Direct estimation is not reliable for such a problem (not many members in each household), and there is a need to borrow strength from the ensemble, as in small area estimation. This is generally done using Bayesian logistic regression. This problem is exasperated because there are numerous households (small areas). A hierarchical Bayesian logistic regression model with random effects is used to capture the
variation of members within and across households. Markov chain Monte Carlo (MCMC) methods require extensive monitoring and the entire
procedure to fit the model is time consuming mostly because there are numerous random effects. Therefore, the joint posterior density of the hyper-parameters is approximated using an integrated nested normal approximation (INNA) via the multiplication rule of probability. This approach provides a sampling-based method that permits fast computation, thereby avoiding very time-consuming MCMC methods from the exact method. In this paper, our main contribution is to obtain an approximate joint posterior density for the hierarchical Bayesian logistic
regression model and to get reasonable estimates and standard errors of small area parameters (e.g., household proportions).
 
The estimation of parameters of the binary logistic regression with random effects is not straight forward due to fact that the likelihood involves multiple integrals. In case of Bayesian analysis, a natural approach to inference in mixed models was proposed by Paulino, Silva
and Achcar (2005). They estimated the random effects, which are treated as parameters in the presence of misclassified data. They also showed that if the posterior distribution is not possible to be obtained analytically, MCMC methods can be used to approximate them. Souza and Migon (2010) proposed that the inference problem can be solved in an easier way if the random effects of the mixed models are distributed as Student-t or finite mixture of normal distributions. Liu and Dey (2008) used prior distributions like skew-normal distribution.  Santos, Loschi and Arellano-Valle (2013) provided different prior and posterior interpretations for the parameters in the logistic regression model with random intercepts when skew normal distribution are assumed to model random effects. 
They obtained the prior distributions for the different parameters and they discussed odds ratio and median odds ratio using skew-normal
distributions for the random effects.  They  concluded that the misspecification of the random effects parameters can give poor estimate. 
Larsen et al. (2000) gave interpretations of the parameters in the logistic regression model with random effects.
Finally, Chen, Ibrahim and Kim (2008) described important properties of logistic regression for a single population 
(no random effects), and showed how to implement Jeffreys's  prior for binomial regression models.

Our work on logistic regression dates back to Nandram (1989) who discussed discrimination between the logit and the complementary log-log
link functions. Nandram (2000) reviewed the paper of Ghosh et al. (1998) on generalized linear models, logistic regression and Poisson regression being two important special cases. Nandram and Erhardt (2005) showed how to analyze binary data with covariates to maintain
conjugacy for both logistic and Poisson regression model. Nandram and Choi (2010) showed how to analyze binary data with covariates
under nonignorable nonresponse. Different from most researchers, Nandram and Choi (2010) showed how to use logistic regression
to obtain propensity scores, an interesting part of their paper, for small area estimation. Roberts, Rao and Kumar (1987) discussed logistic regression for sample survey data (not small area estimation). Nandram and Chen (1996) show how to accelerate the Gibbs sampler
for a model with latent variables introduced earlier by Albert and Chib (1993) for Bayesian probit analysis. 

Albert and Chib (1993) started an innovative stream of research on Bayesian probit analysis, not logistic regression; they agrued that logistic regression is approximately a special case of their probit analysis. However, we now know that this approximation is poor
in the tails and their algorithm  is a poorly mixing Gibbs sampler (Holmes and Held 2006).  For probit analysis, Holmes and Held (2006) showed how to solve this mixing problem  by incoroporating  latent variables and using the block Gibbs sampler (i.e., some variables are drawn simultaneously). Holmes and Held (2006) extended their approach to logistic regression, albeit for a single sample, not for small area estimation as in the
case of numerous small areas that we are studying here. Technically, even for a simple sample, their sampling algorithm  is very complicated using rejection sampling, the Kolmogorov-Smirnov distribution, part of a representation of the standard logistic distribution, and a generalized inverse Gaussian distribution. However, once a user-friendly program is available, the complexity does not matter. Note
that for simple logistic regression (i.e., a single sample), the Metropolis sampler or rejection sampling can be used in a straightforward manner, and this is faster and much simpler than the method of Holmes and Held (2006).  However, it will be extremely difficult to apply their computational techniques in our case simply because the sum of two logistic random variables (we have two error terms, not one) is not another logistic random variable. 

The other side of our application is that there are numerous small areas (households) and MCMC methods cannot handle
them in real time. So our problem can be classified as a ``big data'' problem.
Scott et al. (2013) defined ``big data'' as data that are too big to comfortably process on a single
machine, either because of processor, memory, or disk bottlenecks. They considered consensus 
Monte Carlo methods which split the data to several machines. Communication between large numbers of machines
is expensive (regardless of the amount of data being communicated), so there is a need for
algorithms that perform distributed approximate Bayesian analyses with minimal communication.
Consensus Monte Carlo operates by running a separate Monte Carlo algorithm on each machine, and then averaging 
individual Monte Carlo draws across machines. One of the examples they gave is on a hierarchical Poisson regression
model (very close to logistic regression). But certainly how to split the data is problematic.  
Miroshnikov and Colon (2015) described parallel MCMC methods  for non-Gaussian posterior distributions.
Fortunately,  in survey sampling the design generally uses stratification which is not artificial, and in this case, 
consensus Monte  Carlo may not be needed; it will be a good idea for a large stratum. 

The procedure we use to approximate the posterior density of the parameters of the hierarchical Bayesian logistic
regression model, called the integrated nested normal approximation (INNA), has a closed resemblance to the
integrated nested Laplace approximation (INLA); see Rue, Martino and Chopin (2009). INNA uses a sampling-based 
procedure, that is accommodated by the multiplication rule of probability; currently INLA  is a fairly popular method 
for making approximations in complicated hierarchical Bayesian model. INLA is a promising alternative to MCMC for big data analysis. However, it requires posterior modes and, for numerous small areas, computation 
of modes becomes time-consuming and challenging for logistic regression model or any generalized linear mixed models. Yet 
INLA has found many useful applications. See, for example, Fong, Rue and Wakefield (2010) for an application on Poisson 
regression, and Illian, S\o rbye and Rue (2012) for a realistic application on spatial point pattern data. We note 
that INLA can be problematic especially for logistic and Poisson hierarchical regression models, even if the modes 
can be computed. For example, Ferkingstad and Rue (2015), attempting to improve INLA, used a copula-based correction 
which adds complexity to INLA. For a comparison of INLA and MCMC, the paper by Held, Schr\"{o}dle and Rue (2010) for
cross-validatory predictive checks is interesting. Unfortunately, the computational cost 
of INLA is exponential in the dimension of the parameter space (or hyperparameter space in the case 
of hierarchical models).  

Of course, there are many other approximations in Bayesian statistics some of which can apply directly to logistic
regression. Approximate Bayesian methods (ABC) were introduced in population genetics (e.g., Beaumont, Zhang
and Balding 2002) to deal with intractable likelihood functions and it uses summary statistics. An important
advance was made by Fearhead and Prangle (2012), who obtained a more principled approach to the construction of
summary statistics. Jaakkola and Jordan (2000), Faes, Omerod and Wand (2011) and Omerod and Wand (2010) studied variational Bayes methods.
Variation methods are very complicated, even for the simplest problem, logistic regression without random effects 
(Jaakkola and Jordan 2000), the analysis is not simple. Moreover, the approximate posteriors delivered by variational Bayes give good accuracy for individual marginal distributions, but not for the joint distribution as a whole.
	
We will not use any of these approximations. The nearest to our procedure is INLA that requires posterior modes and it
is computationally costly to run the application we have in mind because of lack of conjugacy, numerical optimization
is needed. Variational Bayes methods are mathematically too complicated and ABC is not accurate enough even though this is
also an active area of research. Instead of finding the posterior modes, INNA finds approximate modes in closed form, facilitated by the empirical logistic transform (Cox ans Snell 1972). Here both the gradient vector (gradient vector is not zero though) and the gradient term is kept in the second order Taylor's series expansion of the posterior distribution of the regression coefficients. So our method does not need posterior modes as in  INLA; this is an enormous saving in computing time, even more so for numerous households. 

The plan of the rest of the paper is as follows. In Section 2, we describe our main contribution about
our approximation to the joint posterior density. In particular, we describe the integrated nested normal
approximation (INNA) and some theoretical results are provided. In Section 3, we present an  illustrative
example using the Nepal Living Standards Survey (NLSS II).  We compare the approximate method with the
exact method, which is presented  in Appendix A. It is worth noting that the word ``exact'' refers to MCMC without further
approximation. In Section 4, we have discussions and two extensions, both of them can be used to accommodate the NLSS II data
better. Additional technical details are given in the appendices.

\section{ Approximate Theory and Method}
\label{S:2}
\vspace*{-5pt}

The method we developed here for many small areas can be applied to any generalized linear model in the same manner. Of course, the specific models will be different. For example, for the model for Poisson regression is different from the model for logistic regression. However, note that for logistic regression model, the unit level (binary data, not
binomial counts) are used and for Poisson regression model the count data are collected at an aggregate level. 
Our application is on the Nepal Living Standards Survey (NLSS II) and we have binary
data (good health versus poor health) for each individual within a household, and these households are within
wards. Our theory applies to individuals within households or individuals within
wards. We note that the number of members in the $i^{th}, i=1,\ldots,L$, household in the population is $N_i$ and all 
household members are sampled, but not all households in a ward are sampled. Our model will hold for all $L$ households
but it is convenient to present the model for the $\ell \leq L$ sampled households.

Our logistic regression model applies to all $L \ge \ell$ population households. We know 
how many households are in each  ward (sampled or nonsampled) but we do not know how many members are
in each nonsampled household. Let $y_{ij}$ and $\xtil_{ij} =  (1,x_{ij1},\dots,x_{ijp-1} )^T$, denote the responses and the $p$ vector of covariates with an intercept ($x_{ij0}=1$).

A standard hierarchical Bayesian logistic regression model is
	$$
	y_{ij}|\betatil,\nu_i  \stackrel{ind}{\sim} \mbox{Bernoulli} \left\{\frac{e^{\xtil_{ij}'\betatil+\nu_i}}{1+e^{\xtil_{ij}'\betatil+\nu_i}} \right\}, j=1,\ldots,n_i,
	$$
	$$
	\nu_i|\delta^2 \stackrel{iid}{\sim}\mbox{Normal}(0,\delta^2), i=1,\ldots,\ell,
	$$
	$$
	\pi(\betatil,\delta^2) \propto\frac{1}{(1+\delta^2)^2}, \delta^2>0.
	$$
	Here, $\nu_i, i=1,\ldots,\ell$, are the random effects and $\betatil= (\beta_0,\beta_1,\dots,\beta_p)^\prime$ are the regression coefficients with $\delta^2$, the variance of the random effects.	
	
For our approximation methods, we will use an equivalent model.	
It is convenient to separate $\betatil$ into $\beta_0$ and $\betatil_{(0)}$, where $\betatil_{(0)}= (\beta_1, \beta_2, \dots, \beta_p)^\prime$. Omitting the  intercept term from the covariate $\xtil_{ij}$, we
have	
	$$
	y_{ij}|\mu_i,\betatil_{(0)}  \stackrel{ind}{\sim} \mbox{Bernoulli} \left\{\frac{e^{\xtil_{ij}'\betatil_{(0)}+\mu_i}}{1+e^{\xtil_{ij}'\betatil_{(0)}+\mu_i}} \right\}, j=1,\ldots,n_i,
	$$
	$$
	\mu_i|\beta_0,\delta^2 \stackrel{iid}{\sim}\mbox{Normal}(\beta_0,\delta^2), i=1,\ldots,\ell,
	$$
	\begin{equation}
	\label{eq:3}
	\pi(\betatil,\delta^2) \propto\frac{1}{(1+\delta^2)^2}, \delta^2>0,
	\end{equation}
	where essentially we have made the transformation $\mu_i = \nu_i + \beta_0,
	 i=1,\ldots,\ell$.
	
	The parameters of interest are the household proportions,
	\begin{equation}
	\label{pint}
	P_{i} = \frac{1}{N_i} \sum_{j=1}^{N_i} I_{ij},  ~~I_{ij} \mid \nu_i, \betatil \stackrel{ind} \sim \mbox{Bernoulli}\{
	\frac{e^{\xtil_{ij}'\betatil+\nu_i}}{1+e^{\xtil_{ij}'\betatil+\nu_i}}\}, i=1,\ldots,\ell.
	\end{equation}
	The $P_i$ give smoothed estimates at the household level and actually predictions
	for the nonsample households.
	A similar formula can be written down for the wards. Because we are not linking
	the census to the NLSS, we do not have the covariates and the number of members in
	each nonsampled households, both being obtained using a Bayesian bootstrap (Rubin 1981) of the 
	original samples.
	
	To develop the approximate methodology, we will work with the no-intercept model. Then, using Bayes' theorem, the joint posterior density for the parameters is
	 \begin{equation}
	 \label{eq:4}
	\pi(\mutil,\betatil,\delta^2|\ytil)\\
	 \propto   \left\{ \prod_{i=1}^\ell  [\prod_{j=1}^{n_i}\frac{e^{(\xtil_{ij}'\betatil_{(0)}+\mu_i)y_{ij}}}{1+e^{\xtil_{ij}'\betatil_{(0)}+\mu_i}} ]
	 [\frac{1}{\sqrt{2\pi\delta^2}}e^{-\frac{(\mu_i-\beta_0)^2}{2\delta^2}} ] \right\}
	\frac{1}{(1+\delta^2)^2}.
	\end{equation}
	The posterior density in (\ref{eq:4}) is a non-standard density, and there are difficulties
	in fitting it using Markov chain Monte Carlo methods, more so when $\ell$ is large. This
	motivates our approximate methods.
	
	We note that given $\betatil,\delta^2,\ytil$, the $\mu_i$ are independent with
	 \begin{equation}
	 \label{eq:cpdv}
	\pi(\mu_i \mid \betatil,\delta^2,\ytil) \propto
	  [\prod_{j=1}^{n_i}\frac{e^{(\xtil_{ij}'\betatil_{(0)}+\mu_i)y_{ij}}}{1+e^{\xtil_{ij}'\betatil_{(0)}+\mu_i}} ]
	 [\frac{1}{\sqrt{2\pi\delta^2}}e^{-\frac{(\mu_i-\beta_0)^2}{2\delta^2}} ], i=1,\ldots,\ell.
	\end{equation}
	For the nonsampled households, $\mu_i \mid \beta_0, \delta^2 \stackrel{iid} \sim \mbox{Normal}(\beta_0, \delta^2), i=1,\ldots,\ell$.
	More importantly with respect to posterior inference about $P_i$, the conditional posterior densities for the $\nu_i$  can be written down easily.

\subsection{Approximation to the Posterior Density}

   In this section we discuss the approximation to the joint posterior density
	in (\ref{eq:4}).

   Let $f(\tautil) = e^{h(\tautil)}$ denote the density of a vector of parameters
	$\tautil$. Let $\gtil$ denote the gradient vector and $H$ the Hessian matrix
	at some point $\tautil^\ast$.

\begin{lemma}
\emph{	Let $h(\tautil)$ be a logconcave density function with the parameter $\tautil$.
	Then, $\tautil$ approximately has a multivariate normal distribution, 
	$$
	\tautil\stackrel{}{\sim}\mbox{Normal} (\tautil^*-H^{-1}\gtil,-H^{-1} ).
	$$}
\end{lemma}	
\begin{proof}
	
	Simply applying  the second-order multivariate Taylor series of $h(\tautil)$ at 
	$\tautil^*$, we have
	$$
	f(\tautil)\approx f(\tautil^*)+(\tautil-\tautil^*)'\gtil+\frac{1}{2}(\tautil-\tautil^*)'H(\tautil-\tautil^*).
	$$
	
We remark that because of logconcavity, $-H$ is positive definite. Also because we are not required to find the mode of $h(\tautil)$, 
$\tautil^\ast$ does not have to be the solution of the gradient vector set to the zero vector. So that the	term involving $\gtil$ is a correction to $\tautil^\ast$.
\end{proof}	


Momentarily, we consider a flat prior $\betatil_{(0)}^\ast$ and the $\mu_i$ (i.e., fiducial inference). That is,
		$$
		y_{ij}|\mu_i,\betatil_{(0)} \stackrel{ind}{\sim} \mbox{Bernoulli} \left\{\frac{e^{\xtil_{ij}'\betatil_{(0)}+\mu_i}}{1+e^{\xtil_{ij}'\betatil_{(0)}+\mu_i}} \right\}, j=1,\ldots,n_i, i=1,\ldots,\ell,
		$$
	$$
	p(\mutil, \betatil_{(0)}) = 1.
	$$
	The joint posterior density is
 \begin{equation}
	 \label{apod}
	\pi(\mutil,\betatil|\ytil)\\
	 \propto\prod_{i=1}^\ell \left\{ \prod_{j=1}^{n_i}\frac{e^{(\xtil_{ij}'\betatil_{(0)}+\mu_i)y_{ij}}}{1+e^{\xtil_{ij}'\betatil_{(0)}+\mu_i}} \right\}.
	\end{equation}
   The logarithm of the joint posterior density (or log-likelihood) is
	$$
	\Delta=h(\tautil)=\sum_{i=1}^\ell \sum_{j=1}^{n_i}\{ (\xtil_{ij}'\betatil_{(0)}+\mu_i)y_{ij} - 
	\log(1+e^{\xtil_{ij}'\betatil_{(0)}+\mu_i})\}.
	$$
  Let $\tautil^\prime=(\mutil^\prime, \betatil_{(0)}^\prime)$. First, we find a 
	convenient point to expand the log-likelihood in a multivariate Taylor's series
	expansion. In Appendix B, we show how to obtain quasi-modes for $\betatil_{(0)}$ and
		$\mu_i, i=1,\ldots,\ell$, of the log-likelihood function.
	
		First, we use the empirical logistic transform $z_i$ to get an estimate of $\mu_i$, where
		$$z_i=\log \left\{\frac{\bar{y}_i+\frac{1}{2n_i}}{1-\bar{y}_i+\frac{1}{2n_i}} \right\}, i=1,\ldots,\ell.$$	
		See Appendix C.
		\par	
		Second, obtain the first derivative of log-likelihood of $\betatil_{(0)}$, use  a first-order Taylor's expansion  
		with $\mu_i$ replaced by $z_i$, and set to zero  to get		
$$
		\betatil_{(0)}^*= \left[ \sum_{i=1}^\ell\sum_{j=1}^{n_i}\xtil_{ij}\xtil_{ij}' \right] ^{-1} \left[ \sum_{i=1}^\ell\sum_{j=1}^{n_i}\xtil_{ij}(y_{ij}-z_i) \right] .
$$ \par
Third, we obtain quasi-modes for the $\mu_i$, a refinement of the $z_i$. We use the log-likelihood of the $\mu_i$ with
the regression coefficients replaced by $\betatil_{(0)}^\ast$, and solve its first derivative for zeros using a first-order
Taylor's series expansion to get
$$
		{\mu_i}^* = \log \left[ \frac{\frac{1}{n_i}\sum _{j=1}^{n_i}e^{-\xtil_{ij}'\betatil_{(0)}^*}}{(1-\bar{y}_i+\frac{1}{2n_i})}\right] , i=1,\ldots,\ell.
$$
Let $\tautil^{\ast \prime}=(\mutil^{\ast \prime}, \betatil_{(0)}^{\ast \prime})$.

Next, we evaluate     				
		$\gtil$ and H at the quasi-modes $\tautil=\tautil^*$,
		$$
		 \gtil=\left(\begin{array}{r}
		\frac{\partial\Delta}{\partial{\mu_1}} \quad\cdots \quad\frac{\partial\Delta}{\partial{\mu_\ell}} \quad\frac{\partial\Delta}{\partial\betatil_{(0)}}
		\end{array}\right)
		^{T}_{\mutil=\mutil^*,\;\betatil_{(0)}=\betatil_{(0)}^*},
		$$
		$$
		H=\begin{pmatrix}
		\frac{\partial^2\Delta}{\partial{\mu_1}^2}&\cdots &0 &\frac{\partial^2\Delta}{\partial{\mu_1}\partial\betatil_{(0)}} \\
		\vdots & \ddots & \vdots & \vdots\\
		0&\cdots &\frac{\partial^2\Delta}{\partial{\mu^2_{\ell}}} &\frac{\partial^2\Delta}{\partial{\mu_{\ell}}\partial\betatil_{(0)}} \\
		\frac{\partial^2\Delta}{\partial{\mu_1}\partial\betatil_{(0)}}&\cdots&\frac{\partial^2\Delta}{\partial{\mu_{\ell}}\partial\betatil_{(0)}} &\frac{\partial^2\Delta}{\partial\betatil^2_{(0)}}
		\end{pmatrix}_{\mutil=\mutil^*,\betatil_{(0)}=\betatil_{(0)}^*}.
		$$
		The partial derivatives can be expressed in terms of response $y_{ij}$ and covariates $\xtil_{ij}$ as
		$$ 
		 \frac{\partial\Delta}{\partial\betatil_{(0)}}=\sum_{i=1}^\ell\sum_{j=1}^{n_i} ( \xtil_{ij} y_{ij}-\frac{\xtil_{ij} e^{\xtil_{ij}'\betatil_{(0)}+\mu_i}}{1+e^{\xtil_{ij}'\betatil_{(0)}+\mu_i}} )
		 $$
		 $$
		 \frac{\partial\Delta}{\partial{\mu_i}}=\sum_{j=1}^{n_i} (y_{ij}-
		 \frac{e^{\xtil_{ij}'\betatil_{(0)}+\mu_i}}{1+e^{\xtil_{ij}'\betatil_{(0)}+\mu_i}} ), i=1,\ldots,\ell,
		$$
		$$ \frac{\partial^2\Delta}{\partial\betatil_{(0)}^2}=-\sum_{i=1}^\ell\sum_{j=1}^{n_i} \frac{\xtil_{ij}\xtil_{ij}'e^{\xtil_{ij}'\betatil_{(0)}+\mu_i}}{(1+e^{\xtil_{ij}'\betatil_{(0)}+\mu_i})^2},
		$$
		$$		 \frac{\partial^2\Delta}{\partial{\mu_i}^2}=-\sum_{j=1}^{n_i}\frac{e^{\xtil_{ij}'\betatil_{(0)}+\mu_i}}
		{(1+e^{\xtil_{ij}'\betatil_{(0)}+\mu_i})^2}, i=1,\ldots,\ell,
		$$
		$$
				 \frac{\partial^2\Delta}{\partial{\mu_i}\partial\mu_{i^\prime}}=0, i \ne i^\prime =1,\ldots,\ell,
		 ~~\frac{\partial^2\Delta}{\partial{\mu_i}\partial\betatil_{(0)}}=-\sum_{j=1}^{n_i}\frac{\xtil_{ij} e^{\xtil_{ij}'\betatil_{(0)}+\mu_i}}{(1+e^{\xtil_{ij}'\betatil_{(0)}+\mu_i})^2}, i=1,\ldots,\ell.
		 $$
		 
		For the convenience of computation, denote $\gtil=\left(\begin{array}{c} \gtil_1\\\gtil_2\end{array}\right)$ and $H=-\left(\begin{array}{c}
		D\quad C'\\C \quad B\end{array}\right),$ where
		$$
		 \gtil_1=\left(\begin{array}{r}
		\frac{\partial\Delta}{\partial{\mu_1}} \cdots \frac{\partial\Delta}{\partial{\mu_\ell}}\end{array}
		\right)^T,
		\gtil_2=\frac{\partial\Delta}{\partial\betatil_{(0)}},
		$$
		$$
		 B=-\frac{\partial^2\Delta}{\partial\betatil_{(0)}^2},
		C=-\left(\begin{array}{r}
		\frac{\partial^2\Delta}{\partial{\mu_1}\partial\betatil_{(0)}}\quad \cdots\quad \frac{\partial^2\Delta}{\partial{\mu_\ell}\partial\betatil_{(0)}}\end{array}\right),
		D=-\left(\begin{array}{r}
		\frac{\partial^2\Delta}{\partial{\mu_1}^2} \quad\cdots \quad0\\
		\colon \quad\ddots \quad\colon\\
		0 \quad\cdots \quad\frac{\partial^2\Delta}{\partial{\mu_\ell}^2}\\
		\end{array}\right).
		$$
		Note that D is a diagonal matrix. 
		$$ \mbox{Let}~
		-H^{-1}=\left(\begin{array}{c}
		D \quad C'\\C \quad B
		\end{array}\right)^{-1}=\left(\begin{array}{c}
		E \quad F'\\F \quad G
		\end{array}\right), ~\mbox{where} $$
		 $$E=D^{-1}+D^{-1}C'(B-CD^{-1}C')^{-1}CD^{-1},
		 F=-(B-CD^{-1}C')^{-1}CD^{-1},
		 G=(B-CD^{-1}C')^{-1}.$$

\begin{lemma}
	
  \emph{ Assuming that the design matrix is full-rank and $0<\sum_{j=1}^{n_i} y_{ij} <n_i, i=1,\ldots,\ell$, the posterior density, $\tautil|\ytil$ in (\ref{apod}), is logconcave.}
\end{lemma}	
	
\begin{proof}
If $0<\sum_{j=1}^{n_i} y_{ij} <n_i, i=1,\ldots,\ell$, there are solutions to the gradient vector set
to zero. 

Let $p_{ij} =  \frac{e^{\xtil_{ij}'\betatil_{(0)}+\mu_i}}{1+e^{\xtil_{ij}'\betatil_{(0)}+\mu_i}}, j=1,\ldots,n_i, i=1,\ldots,\ell$. Then,
$A$, $B$ and $C$ of the  negative Hessian matrix can be written as,
	
		$$ B = \frac{\partial^2\Delta}{\partial\betatil_{(0)}^2}=\sum_{i=1}^\ell\sum_{j=1}^{n_i} 
		p_{ij}(1-p_{ij})\xtil_{ij}\xtil_{ij}',
		$$
		$$	
			D = \mbox{diagonal}(d_i, i=1,\ldots,\ell), d_i = \frac{\partial^2\Delta}{\partial{\mu_i}^2}
			  =\sum_{j=1}^{n_i}p_{ij}(1-p_{ij}), i=1,\ldots,\ell
		$$
		$$
		C = (\mathunderaccent\tilde{c}_i),
		\mathunderaccent\tilde{c}_i		=
		 ~~\frac{\partial^2\Delta}{\partial{\mu_i}\partial\betatil_{(0)}}=\sum_{j=1}^{n_i} p_{ij}(1-p_{ij})\xtil_{ij}, i=1,\ldots,\ell.
		$$
 
     Clearly, $D$ is positive definite. Thus, to show that $-H$ is positive definite, we show that
		its Schur complement of $D$, $S = B-C D^{-1} C^\prime$, is positive definite (e.g., see Boyd 
		and Vandenberghe 2004). Let $\omega_{ij} = p_{ij}(1-p_{ij})/\sum_{j=1}^{n_i} \{p_{ij}(1-p_{ij})\},
		j=1,\ldots,n_i, i=1,\ldots,\ell$. Then, the Schur complement is
		$$
		S = \sum_{i=1}^\ell \sum_{j=1}^{n_i} p_{ij}(1-p_{ij}) \sum_{j=1}^{n_i} \omega_{ij} 
		\xtil_{ij} \xtil_{ij}^\prime - \sum_{i}^\ell \sum_{j=1}^{n_i} p_{ij}(1-p_{ij}) 
		\sum_{j=1}^{n_i}   \omega_{ij} \xtil_{ij}    \sum_{j=1}^{n_i}  \sum_{j=1}^{n_i} 
		\omega_{ij} \xtil_{ij}^\prime.
		$$ 
		It is now easy to show that
		$$
		S = \sum_{i=1}^\ell \sum_{j=1}^{n_i} \omega_{ij}(\xtil_{ij}-\sum_{j=1}^{n_i} \omega_{ij} \xtil_{ij})
		(\xtil_{ij}-\sum_{j=1}^{n_i} \omega_{ij} \xtil_{ij})^\prime.
		$$
		Therefore,  $-H$ is positive definite, and $\tautil|\ytil$ is logconcave.
\end{proof}		

    We next establish the first key result presented in Theorem 2.1.
 \begin{theorem}
 \emph{Assuming that the design matrix is full-rank and $0<\sum_{j=1}^{n_i} y_{ij} <n_i$, the posterior density, $\tautil|\ytil$ in (\ref{apod}) is approximately
   a multivariate normal density, and
   $$
		 \mutil|\betatil_{(0)},\ytil\stackrel{}{\sim}
		\mbox{Normal}\{\mutil_{\mu}-D^{-1}C'(\betatil_{(0)}-\mutil_{\beta}),D^{-1}\} ~\mbox{and}~
			\betatil_{(0)}|\ytil\stackrel{}{\sim}\mbox{Normal}\{\mutil_{\beta},G\},
		$$ 
where
$$
\mutil_{\mu} = \mutil^*+E\gtil_1+F'\gtil_2   ~\mbox{and}~ \mutil_{\beta} = \betatil_{(0)}^*+F\gtil_1+G\gtil_2.
$$}
\end{theorem} 
 
\begin{proof}
		
		By Lemma 2.2 the posterior density is logconcave. By Lemma 2.1 the posterior density is approximately
		a multivariate normal density. We provide the approximate mean  and variance to completely specify the
		multivariate normal density.
		
		By Lemma 2.1, evaluating all appropriate quantities at $\tautil^\ast$, the posterior mean is
		$$ 
		\left(\begin{array}{c}\mutil_{\mu}\\\mutil_{\beta}\end{array}\right) = 	\tautil^*-H^{-1}\gtil	
		=\left(\begin{array}{c}
		\mutil^*\\\betatil_{(0)}^*
		\end{array}\right)+\left(\begin{array}{c}E\quad F'\\F\quad G\end{array}\right)
		\left(\begin{array}{c}\gtil_1\\\gtil_2\end{array}\right)
		=\left(\begin{array}{c}\mutil^*+E\gtil_1+F'\gtil_2\\\betatil_{(0)}^*+F\gtil_1+G\gtil_2\end{array}\right).
		 $$
     Also, by Lemma 1, the posterior variance is 		
		$$
		-H^{-1}=\left(\begin{array}{c}
		D \quad C'\\C \quad B
		\end{array}\right)^{-1}=\left(\begin{array}{c}
		E \quad F'\\F \quad G
		\end{array}\right).
		$$
		That is, the approximate joint posterior density is
		$$
		\left(\begin{array}{c}
		\mutil\\\betatil_{(0)}
		\end{array}\right)|\ytil\stackrel{}{\sim} \mbox{Normal} \left\{\left(\begin{array}{c}
		\mutil_{\mu}\\\mutil_{\beta}
		\end{array}\right),  \left(\begin{array}{c}
		E \quad F'\\F \quad G
		\end{array}\right) \right\}.
		$$

		Finally, using a standard property of the multivariate normal density,
		it follows that approximately,
		$$
		 \mutil|\betatil_{(0)},\ytil\stackrel{}{\sim}
		\mbox{Normal}\{\mutil_{\mu}-D^{-1}C'(\betatil_{(0)}-\mutil_{\beta}),D^{-1}\}
		~\mbox{and}~
			\betatil_{(0)}|\ytil\stackrel{}{\sim}\mbox{Normal}\{\mutil_{\beta},G\}.
		$$
\end{proof}
\subsection{Integrated Nested Normal Approximation} 
	
	In this section, we obtain the integrated nested normal approximation (INNA).
 INNA, which does not require posterior modes, is competitive to the integrated nested
 Laplace approximation (INLA) that requires posterior modes.

  Next, using the normal priors for the $\mu_i$ and Theorem 1, we have the following approximate hierarchical Bayesian regression model,
	$$
	\mutil|\betatil_{(0)},\ytil\stackrel{}{\sim}\mbox{Normal}\{\mutil_{\mu}-D^{-1}C'(\betatil_{(0)}-\mutil_{\beta}),D^{-1}\}
	 $$
	$$
	\betatil_{(0)}|\ytil\stackrel{}{\sim}\mbox{Normal}\{\mutil_{\beta},G\}
	 $$
	$$
	\mutil|\beta_0,\delta^2\sim{\mbox{Normal}\{\beta_0\jtil,\delta^2I\}}
	 $$
	\begin{equation}
	\label{amod}
	\pi(\beta_0,\betatil_{(0)},\delta^2)\propto\frac{1}{(1+\delta^2)^2}, \delta^2>0,
	\end{equation}
	where $\jtil$ is a vector of ones.
	
	Then, using Bayes' theorem again, the approximate joint posterior density for the parameters $\mutil,\betatil$ and $\delta^2$ is
	$$
\pi_a(\mutil,\betatil,\delta^2|\ytil) \propto
e^{-\frac{1}{2}\left\{\left[\mutil-\left(\mutil_{\mu}-D^{-1}C'(\betatil_{(0)}-\mutil_{\beta})\right)\right]'D\left[\mutil-\left(\mutil_{\mu}-D^{-1}C'(\betatil_{(0)}-\mutil_{\beta})\right)\right]+\left[\mutil-\beta_0\jtil\right]'(\delta^2I)^{-1}\left[\mutil-\beta_0\jtil\right]\right\}}
	$$
	\begin{equation}
	\label{eq:14}
		\times\frac{|D|^{1/2}}{{|\delta^2I|}^{1/2}|G|^{1/2}}\times\frac{1}{(1+\delta^2)^2}\times e^{-\frac{1}{2}\left[\betatil_{(0)}-\mutil_{\beta}\right]'G^{-1}\left[\betatil_{(0)}-\mutil_{\beta}\right]}.
	\end{equation}
	
	Next, we state the main result of the paper in Theorem 2.2.
	
\begin{theorem}
  \emph{	Using the multiplication rule, the joint posterior density, 
	$\pi(\mutil,\betatil,\delta^2 \mid \ytil)$ in (\ref{eq:14}), can be approximated by
	$$
	\pi_a(\mutil,\betatil,\delta^2 \mid \ytil) = \pi_a(\mutil \mid \betatil,\delta^2, \ytil)
	\pi_a(\betatil \mid \delta^2, \ytil)
	\pi_a(\delta^2 \mid \ytil),
	$$
	where the three densities on the right-hand side are to be determined.}
\end{theorem}

\begin{proof}
First, it can be shown that		
		\begin{equation}		\mutil|\betatil,\delta^2,\ytil\sim{\mbox{Normal}\left\{(D+\frac{1}{\delta^2}I)^{-1}(D\mutil_{\mu}-C'(\betatil_{(0)}-\mutil_{\beta})+\frac{1}{\delta^2}\beta_0\jtil),(D+\frac{1}{\delta^2}I)^{-1}\right\}},
		\end{equation}
which is $\pi_a(\mutil \mid \betatil,\delta^2, \ytil)$.
Because $(D+\frac{1}{\delta^2}I)$ is diagonal, given $\betatil,\delta^2,\ytil$, the $\mu_i$ are independent. This is an 
important result because  parallel computation can be done for $\mu_i$, which accommodates time-consuming and massive storage challenges in big data analysis. This result holds for the exact conditional posterior density of the 
$\mu_i$.
	
	Second, because $\mutil$ has a multivariate normal distribution, we can integrate out $\mutil$ from the joint posterior density $\pi_a(\mutil, \betatil,\delta^2|\ytil)$ to get the joint posterior density of $\betatil$ and $\delta^2$,
	
	$
	\pi_a(\betatil,\delta^2|\ytil)
	\propto
	\frac{|D|^{1/2}}{{|D+\frac{1}{\delta^2}I|^{1/2}|\delta^2I|}^{1/2}|G|^{1/2}}\times\frac{1}{(1+\delta^2)^2} \times e^{-\frac{1}{2} \left[\betatil_{(0)}-\mutil_{\beta}\right]'G^{-1}\left[\betatil_{(0)}-\mutil_{\beta}\right] }
	$
	\begin{equation}
	\label{eq:16}
	 \times e^{-\frac{1}{2}\left\{\left[\mutil_{\mu}-D^{-1}C'(\betatil_{(0)}-\mutil_{\beta})-\beta_0\jtil\right]'
(D^{-1}+\delta^2I)^{-1}\left[\mutil_{\mu}-D^{-1}C'(\betatil_{(0)}-\mutil_{\beta})-\beta_0\jtil\right]\right\}}
	 \end{equation}
	and
	
	$
	\pi(\betatil|\delta^2,\ytil) \propto e^{-\frac{1}{2} \left[\betatil_{(0)}-\mutil_{\beta}\right]'G^{-1}\left[\betatil_{(0)}-\mutil_{\beta}\right] }
	$
	\begin{equation}
	\label{eq:17}	 \times e^{-\frac{1}{2}\left\{\left[\mutil_{\mu}-D^{-1}C'(\betatil_{(0)}-\mutil_{\beta})-\beta_0\jtil\right]'(D^{-1}+\delta^2I)^{-1}\left[\mutil_{\mu}-D^{-1}C'(\betatil_{(0)}-\mutil_{\beta})-\beta_0\jtil\right]\right\}}.
	 \end{equation}

Let	
		$$
	\Delta_{(0)}=CD^{-1}(D^{-1}+\delta^2I)^{-1}D^{-1}C'+G^{-1},
	$$
	$$
	\delta^2_0=\jtil'(D^{-1}+\delta^2I)^{-1}\jtil,
	$$
	$$
	\gmtil=CD^{-1}(D^{-1}+\delta^2I)^{-1}\jtil,
	 $$
			$$
	\left(\begin{array}{c} \omega_0\\\ometil_{(0)}\end{array}\right)
	=\left(\begin{array}{c}
	\delta^2_0 \quad \gmtil'\\
	\gmtil \quad \Delta_{(0)}
	\end{array}\right)^{-1}
	\left(\begin{array}{c}
	(\mutil_{\mu}+D^{-1}C'\mutil_{\beta})'(D^{-1}+\delta^2I)^{-1}\jtil\\
	(\mutil_{\mu}+D^{-1}C'\mutil_{\beta})'(D^{-1}+\delta^2I)D^{-1}C'+\mutil_{\beta}'G^{-1}
	\end{array}\right).
	 $$
After extensive algebraic manipulation, we can show that
$\betatil|\delta^2,\ytil$ has an approximate multivariate normal density,
	\begin{equation}
		\left(\begin{array}{c}
		\beta_0\\\betatil_{(0)}
		\end{array}\right)|\delta^2,\ytil\sim \mbox{Normal}\left\{
		\left(\begin{array}{c}
		\omega_0\\\ometil_{(0)}
		\end{array}\right),
		\left(\begin{array}{c}
		\delta^2_0 \quad\gmtil'\\
		\gmtil\quad\Delta_{(0)}
		\end{array}\right)^{-1}\right\},
		\label{eq:beta}
	 \end{equation}
	which is denoted by $\pi_a(\betatil \mid\delta^2, \ytil)$.
	
	Third, because the  approximate conditional distribution of $\betatil|\delta^2, \ytil$ is a multivariate normal distribution, we can integrate out $\betatil$ from the joint density of $\pi_a(\betatil,\delta^2 \mid \ytil)$ in (\ref{eq:16}) to get the posterior density,
	$$
	\pi_a(\delta^2|\ytil)
	\propto
	\frac{1}{|\delta^2D+I|^{1/2}}\left|\left(\begin{array}{c}
	\delta^2_0 \quad \gmtil'\\
	\gmtil \quad \Delta_{(0)}
	\end{array}\right)\right|^{-\frac{1}{2}}\times\frac{1}{(1+\delta^2)^2}
	$$
	$$
	\times
	e^{-\frac{1}{2}\left[(\mutil_{\mu}+D^{-1}C'\mutil_{\beta})'(D^{-1}+\delta^2I)^{-1}(\mutil_{\mu}+D^{-1}C'\mutil_{\beta})
		+\mutil_{\beta}'G^{-1}\mutil_{\beta}
		-\wtil'\left(\left(\begin{array}{c}
		\delta^2_0 \quad \gmtil'\\
		\gmtil \quad \Delta_{(0)}
		\end{array}\right)\right)\wtil\right]},
	 $$
	 where $\wtil^\prime = (\omega_0, \ometil_{(0)}^\prime)$.
	
\end{proof}	
	Our INNA method is sampling based and it uses random samples. Samples are obtained by first
	drawing from $\pi_a(\delta^2|\ytil)$, then $\pi_a(\betatil \mid\delta^2, \ytil)$ and finally
	$\pi_a(\mutil \mid \betatil,\delta^2, \ytil)$, the first two densities use standard draws, but
	the third is a nonstandard density.

	Since $0<\delta^2 <\infty$, we make a transformation to $\eta=\frac{1}{1+\delta^2}$ so that $0<\eta <1$. Then,
	posterior density, $\pi(\eta|\ytil)$, is
	
	$
	\pi(\eta|\ytil)\propto
	\left\{\frac{1}{|\delta^2D+I|^{1/2}}\left|\left(\begin{array}{c}
	\delta^2_0 \quad \gmtil'\\
	\gmtil \quad \Delta_{(0)}
	\end{array}\right)\right|^{-\frac{1}{2}}\right\}_{\delta^2=\frac{1-\eta}{\eta}} 
	$
	\begin{equation}
	\label{eq:22}
	\times
	\left\{e^{-\frac{1}{2}\left[(\mutil_{\mu}+D^{-1}C'\mutil_{\beta})'(D^{-1}+\delta^2I)^{-1}(\mutil_{\mu}+D^{-1}C'\mutil_{\beta})
		+\mutil_{\beta}'G^{-1}\mutil_{\beta}
		-\wtil'\left(\left(\begin{array}{c}
		\delta^2_0 \quad \gmtil'\\
		\gmtil \quad \Delta_{(0)}
		\end{array}\right)\right)\wtil\right]}\right\}_{\delta^2=\frac{1-\eta}{\eta}}.
	 \end{equation}
	The grid method is used to sample $\eta$.

	%
			%
	
	INNA simply uses the multiplication rule to get samples from the approximate joint posterior density, which is
	\begin{equation}
	\label{eq:apos}
	\pi_a(\mutil,\betatil, \delta^2 \mid \ytil) =
	\pi_a(\mutil, \mid \betatil, \delta^2, \ytil)
	\pi_a(\betatil \mid \delta^2, \ytil)
	\pi_a(\delta^2 \mid \ytil),
	\end{equation}
	where $\pi_a(\mutil, \mid \betatil, \delta^2, \ytil)$ is given in (\ref{eq:cpdv}), 
	$\pi_a(\betatil, \delta^2 \mid \ytil)$ is given in (\ref{eq:beta}) and $\pi_a(\delta^2 \mid \ytil)$ is given in
	(\ref{eq:22}) after re-transformation.
	Equation (\ref{eq:apos}) is the basis of our INNA approximation. This is simply the multiplication rule of probability;
	simply draw $\eta$ from (\ref{eq:22}) and retransform to $\delta^2$ to get  $\pi_a(\delta^2 \mid \ytil)$,
	draw $\betatil$ from $\pi_a(\betatil \mid \delta^2, \ytil)$ (multivariate normal and $\mutil$ from
	$\pi_a(\mutil, \mid \betatil, \delta^2, \ytil)$. Use a Metropolis algorithm with an approximate normal proposal density  to draw samples of $\nu_i$ independently given $\betatil$, $\delta^2$ and data.  We run each Metropolis 100 times and picked the last one. 
If the Metropolis step fails (jumping rate is not in $(.25, .50)$, we use a grid method instead.  Parallel computing can also be used in this latter step. This is performed in the same manner for the exact method. For our application with 3,912 households, this latter step  runs very fast. Of course, with much larger number of households, the computing time will be substantial, but now parallel
computing is available. 

\subsection{Comparison of the Two Methods}

As a summary, we compare the approximate and exact methods. The exact method is given in Appendix A.

First, we note that the exact method actually uses the approximate method. We
use a Metropolis step with $\pi_a(\betatil,\delta^2|\ytil)$, obtained from the approximate method.
We fit a multivariate Student's $t$ distribution with $\eta$ degrees of freedom to the iterates, $(\betatil, \log(\delta^2))$
from the approximate method as the proposal density in the Metropolis step. It is standard to tune the Metropolis step 
by varying $\eta$.

We present two differences between the two methods. First, both methods are sampling based; the approximate method
implements random samples and the exact method a Markov chain. 
Second, $\pi_a(\betatil, \delta^2 \mid \ytil)$ is used for the approximate method
and a Metropolis step is used for $\pi(\betatil, \delta^2 \mid \ytil)$.
It is this Metropolis step that is time-consuming (20 minutes versus 20 seconds); for a million households, we
can prorate this time (80 hours versus 80 minutes), enormous savings.

We present three similarities between the two methods.
First, $\pi(\nu_i \mid \betatil, \delta^2, \ytil), i=1,\ldots, \ell$, are drawn
the same way using a Metropolis step with proposal density $\pi_a(\nu_i \mid \betatil, \delta^2, \ytil), i=1,\ldots, \ell$.
Grids are used when the Metropolis step fails. Parallel computing can be done easily in both cases.
Second, for $\pi(\nu_i \mid \betatil, \delta^2, \ytil), i=\ell+1,\ldots,L$, are normal densities
($\ell = 3,912, ~L=60,262$). Third, for prediction two Bayesian bootstraps are used to get the nonsampled household 
sizes and the nonsampled covariates ($\approx$ two million people). This is done within wards.

\section{Illustrative Example}
\label{S:3}

In Section 3.1, we briefly describe the  Nepal Living Standards Survey (NLSS II) and in Section 3.2, we use the health status
data with five covariates to compare our approximate Bayesian logistic regression method with the exact one.

\subsection{Nepal Living Standards Survey}

We use data from the Nepal Living Standards Survey (NLSS II, Central Bureau of Statistics, 2003-2004) to illustrate INNA with logistic regression.  NLSS is a national household survey in Nepal, actually population based (i.e., interviews are done for individual household members). NLSS follows the World Bank's Living Standards Measurement Survey  methodology with a two-stage stratified sampling scheme, which has been successfully applied in many parts of the world. It is an integrated survey which covers samples from the whole country and runs throughout the year. The main objective of the NLSS is to collect data from Nepalese households and provide information to monitor progress in national living standards. The NLSS gathers information on a variety of areas. It has collected data on demographics, housing, education, health, fertility, employment, income, agricultural activity, consumption, and various other areas. NLSS has records for 20,263 individuals from 3,912 households from a 326 wards (or psus) from a population of 60,262 households and about two million Nepalese. We choose the binary variable, health status, from the health section of the questionnaire. There are hundreds of variables and we have selected five of the most important ones that can explain health status. We use health status as the binary variable with a selection of five pertinent covariates.

Health status is covered in Section 8 of the questionnaire. This section collected information on chronic and acute illnesses, uses of medical facilities, expenditures on them and health status. Health status questionnaire is asked for every individual that was covered in the survey. The health status questionnaire has four options. For our purpose we make it a binary variable,
good health or poor health.
 
In the NLSS II, Nepal is divided into wards (psu's) and within each ward there are a number of households. The sampling design of the  NLSS II used two-stage stratified sampling.  A sample of psu's was selected using PPS sampling and then twelve households were systematically selected from each ward. Thus, households have equal probability of selection. But while individuals in a household has equal probability of selection, these probabilities will vary with the size of the households. That is, over the entire Nepal, each individual has unequal probability of selection. 

We choose five relevant covariates which can influence health status from the same NLSS survey for the integrated nested normal approximation (INNA) logistic regression method. They are age, nativity, sex, area and religion. We created binary variables nativity (Indigenous = 1, Non-indigenous = 0), religion ((Hindu = 1, Non-Hindu = 0), sex (Male = 1, Female = 0) and area  (Urban = 1, Rural = 0). We standardize
age. Older age and child age are more vulnerable than younger age. Indigenous people can have different health status from migrated people. Similarly, health status of urban and rural citizens could be different.  Interest is on the smoothed proportion of household members in good health (binary). The NLSS 2003-2004 sample consists 3,912 households (roughly 20 thousand people) from  a population from 326 wards with 60,221 households. From Census 2001, the population of Nepal consists of 35,000 wards from 4 million households with roughly 20 million people. We will predict the finite population proportions of household members in good health based for the 60,221 households, not the entire Nepal. This can be done using the same methodology  because Nepal consists of six strata.

We  use a Bayesian hierarchical logistic regression model to predict all nonsampled households in the sampled wards.
To obtain smoothed estimates for the sampled households, we also predict these. We use Bayesian bootstraps (Rubin 1981) for unknown household sizes and nonsampled covariates; the bootstrapping is done within wards. The 2001 Census can potentially provide these two pieces of information, but there is a mis-match between the households in the census and the NLSS (a record linkage can be performed). We note, however, that there is linkage between the wards, but this information is not useful to household estimates.
  
\subsection{Numerical Comparisons}

We predict the household proportions of members in good health for 60,221 households. This analysis is based on 3,912 sample
households from 326 wards (PSUs). Our primary purpose is to compare the approximate method with the exact method when there 
are random effects at the household level. We want to show that one can safely use the approximate method to save computational 
time. Our secondary purpose is to compare the exact methods when there are random effects at the household level only and random
effects at the ward level. We note that there are only 326 wards, considerably less than the number of households.
In this case, a data analyst might try to save computational time by only using the random effects at the ward level, but this
procedure is not sensible.

First, in Figures  \ref{FG:poPM}, \ref{FG:poPSD} and  \ref{FG:poCV}, we compare respectively the posterior means (PMs),
posterior standard deviations (PSDs) and  posterior coefficient of variations (CVs) as our primary purpose. The PMs
are very closed, followed by the PSDs and then the CVs. For the PMs, the points lie very closely on the $45^o$ 
straight line through the origin. This is similar for the PSDs and CVs. However, the plot of the
PSD is  a bit thicker and the plot for the CVs has larger CVs more spread out. But overall, these approximations are 
quite acceptable to data analysts, scientists and engineers. We also show the posterior densities (PDs) of the 
hyperparameters. In Figure \ref{FG:podb} we plot the PDs of $\delta^2$ and $\beta_0$ and in Figure \ref{FG:pobeta}
we plot the remaining five $\beta$s. We can see that they are all unimodal, with the modes about the same, but, 
as expected,  the spread is a bit larger for the exact method. But these differences are not alarming. However, these 
differences are much smaller when inference is made about the household proportions. 

We also compare the approximate method versus the exact method when random effects are used only at the ward level.
Again, the PMs are very good; see Figure \ref{FG:poPMp}. The PSDs are more spread out (Figure \ref{FG:poPSDp})
and the plot of the CVs (Figure \ref{FG:poCVp}) is a similar to the household random effects. Again, the approximate method 
and the exact method are reasonably closed, of course, not as closed as for the household level random effects.

Finally, we compare the exact method at the ward level to the exact method at the household level. Figures 
\ref{FG:poPMe}, \ref{FG:poPSDe} and \ref{FG:poCVe} show that the comparisons are rather poor; there are no clear
indications of $45^o$ straight line through the origin. We have examined this further in Tables \ref{TAB:PMcat},
\ref{TAB:PScat} and \ref{TAB:PCVcat}, where we cross-classified the PMs in $0-.2$, $.2-.4$, $.4-.8$, $.8-1$,
the PSDs in $.0-.1$, $.1-.2$, $.2-.3$, $.3-.4$, $.4-.5$, and the CVs in $.02-.05$, $.05-.10$, $.10-.25$, $.25-.50$, 
$>.50$. For all three tables  a majority of the  points fall along the north-west to south-east diagonal, but there 
are many points off the diagonals. Thus, one must have random effects at the household level.

We conclude that the approximation at the household level is reasonable. The approximation is desirable
because one can perform the computations in real time.  One should not use random effects only at the
ward level to cut corners. This can be misleading.

\section{Concluding Remarks}

We make three statistical comments.
First,  the approximate method is necessary when there are a large number (millions) of households (clusters or areas).
Second, it is difficult to use the census data effectively but it is desirable (matching problem). Third,
it is possible to obtain similar approximations for spatial priors and Dirichlet process priors (under investigation)       

We make four computational comments. First, parallel computing is needed for big data (numerous small areas).
Second, consensus Monte Carlo method, although problematic, is needed for large data (storage problems).
For the NLSS survey, stratification already exists; Nepal has six strata and our INNA procedure  can be applied in 
each stratum in parallel. Third, MCMC (not INLA) methods are useful for small datasets; good approximations are needed
for large datasets (big data). Finally, INNA is potentially useful because modes are not required (INLA needs modes).      

We mention two possible extensions. The first extension is about survey weights. The second extension
is to a sub-area model.

PPS sampling is used in the first-stage of the survey design. Thus, there are survey weights (design, not adjusted 
weights). All households (each member) in a psu has the same weight. So we can proceed in one of the two ways in our 
analysis. First, we can use an adjusted logistic likelihood incorporating the survey weights. We discuss a single area, 
then we show how to extend it to small area problems. Let $\omega_i, i=1,\ldots,n$, denote the survey weights for
sampling from a single area. Let
$$
n \geq n_e = (\sum_{i=1}^n \omega_i)^2/(\sum_{i=1}^n \omega_i^2),
~~\tilde{\omega}_i = n_e \frac{\omega_i}{\sum_{i=1}^n \omega_i}, i=1,\ldots,n;
$$
see Potthof, Woodbury and Manton (1992) for pioneering work on equivalent sample sizes.\\
For $(y_i, \tilde{\omega}_i,\mathunderaccent\tilde{x}_i), i=1,\ldots,n$, we have
$$
p(y_i \mid \mathunderaccent\tilde{\beta}) \propto 
\left\{ \frac{e^{y_i \mathunderaccent\tilde{x}_i^\prime \mathunderaccent\tilde{\beta}}}
{1+e^{\mathunderaccent\tilde{x}_i^\prime \mathunderaccent\tilde{\beta}}}
\right\}^{\tilde{\omega}_{i}}, y_i = 0, 1, i=1,\ldots,n,
$$
$$
p(y_i \mid \mathunderaccent\tilde{\beta}) =
\frac{e^{y_i \mathunderaccent\tilde{\tilde{x}}_i^\prime \mathunderaccent\tilde{\beta}}}
{1+e^{\mathunderaccent\tilde{\tilde{x}}_i^\prime \mathunderaccent\tilde{\beta}}},
~y_i = 0, 1,$$
$$
 \mathunderaccent\tilde{\tilde{x}}_i = \tilde{\omega}_i \mathunderaccent\tilde{x}_i,  i=1,\ldots,n.
$$
For small areas, with
$(y_{ij}, \tilde{\omega}_{ij},\mathunderaccent\tilde{x}_{ij}), ~j=1,\ldots,n_i, ~i=1,\ldots,\ell$, we have
$$
p(y_{ij} \mid \mathunderaccent\tilde{\beta}, {\nu}_i) =
\frac{e^{y_{ij} (\mathunderaccent\tilde{\tilde{x}}_{ij}^\prime \mathunderaccent\tilde{\beta}+\nu_i)}}
{1+e^{\mathunderaccent\tilde{\tilde{x}}_{ij}^\prime \mathunderaccent\tilde{\beta}+\nu_i}},
~y_{ij} = 0, 1,
$$
$$
 \mathunderaccent\tilde{\tilde{x}}_{ij} = \tilde{\omega}_{ij} \mathunderaccent\tilde{x}_{ij},
~j=1,\ldots,n_i, ~i=1,\ldots,\ell.
$$

Our second extension is to sub-area model. One example already discussed in the literature
is that of Torabi and Rao (2014), who extended the Fay-Herriot model (Fay and Herriot 1979), not for logistic
regression. For our problem, the areas are the wards and sub-areas are the households.
Let $i=1,\ldots,\ell$, denote the areas and $j=1,\ldots,n_{i}$, denote the sub-areas.
We assume that
	$$ 
	y_{ijk}|\betatil,\nu_i, \mu_{ij} \stackrel{ind}{\sim} \mbox{Bernoulli}
\left\{ \frac{e^{\xtil_{ijk}'\betatil+\nu_i + \mu_{ij}}}{1+e^{\xtil_{ijk}'\betatil+\nu_i+\mu_{ij}}}\right\}, k=1,\ldots,m_{ij}, 
	$$
	$$
	\mu_{ij}| \sigma^2 \stackrel{iid} \sim \mbox{Normal}(0,\sigma^2), j=1,\ldots,n_i,
	$$
	$$
	\nu_i|\delta^2\stackrel{iid}{\sim}\mbox{Normal}(0,\delta^2), i=1,\ldots,\ell,
	$$
  $$
	\pi(\betatil,\delta^2, \sigma^2) \propto \frac{1}{(1+\sigma^2)^2}  \frac{1}{(1+\delta^2)^2},
	\sigma^2>0, \delta^2>0.
	$$
For logistic regression, this research is currently in progress.

\begin{appendices}
\section{Exact Method for Logistic Regression}
\def\theequation{A.\arabic{equation}}
\setcounter{equation}{0}

Recall the Bayesian logistic model with covariates that we worked on with INNA method
$$ 
y_{ij}|\mu_i,\betatil_{(0)} \stackrel{ind}{\sim} \mbox{Bernoulli}\left\{ \frac{e^{\xtil_{ij}'\betatil_{(0)}+\mu_i}}{1+e^{\xtil_{ij}'\betatil_{(0)}+\mu_i}}\right\},
$$
$$
\mu_i|\beta_0,\delta^2\stackrel{iid}{\sim}\mbox{Normal}(\beta_0,\delta^2),\\
$$
\begin{equation}
\pi(\betatil,\delta^2)\propto\frac{1}{(1+\delta^2)^2},
\delta^2>0, i=1,...,\ell, j=1,...,n_i.
\end{equation}
According to Bayes' theorem, the joint posterior density of the parameters $(\mutil,\betatil,\delta^2|\ytil)$ is
$$
\pi(\mutil,\betatil,\delta^2|\ytil)\\
\propto \pi(\ytil|\mutil,\betatil_{(0)}) \times \pi(\mutil|\beta_0,\delta^2)\times \pi(\betatil,\delta^2)
$$
$$	\propto\prod_{i=1}^\ell\left\{\left[\prod_{j=1}^{n_i}\frac{e^{(\xtil_{ij}'\betatil_{(0)}+\mu_i)y_{ij}}}{1+e^{\xtil_{ij}'\betatil_{(0)}+\mu_i}}\right]
\left[\frac{1}{\sqrt{2\pi\delta^2}}e^{-\frac{(\mu_i-\beta_0)^2}{2\delta^2}}\right]\right\}
\frac{1}{(1+\delta^2)^2}.
$$

\begin{theorem}
	\emph{The joint posterior density, $\pi(\mutil,\betatil,\delta^2|\ytil)$, is proper provided
		that the design matrix is full rank and $0<\sum_{j=1}^{n_i} y_{ij} < n_i, i=1,\ldots,\ell$.}
\end{theorem}	

\begin{proof}
	With a flat prior on the $\mu_i$ and $\betatil$, the same argument as in Lemma 2.2 gives
	logconcavity of the joint posterior density. Putting a logconcave prior on the $\mu_i$ does
	not change the logconavity of $\pi(\betatil,\mutil \mid \delta^2,\ytil)$ because the product
	of two logconcave densities is another logconcave density. In addition, logconcave densities
	have sub-exponential tails and their moment generating functions exist (see Dharmadhikari and Joag-Dev 1988). That is,
	$\int \pi(\mutil,\betatil \mid \delta^2, \ytil) d\betatil d\mutil = a(\delta^2)$
	finite for all $\delta^2$. Therefore, $\int a(\delta^2) \pi(\delta^2) d\delta^2 < \infty$ as
	long as $\pi(\delta^2)$ is proper as for $\pi(\delta^2) = 1/(1+\delta^2)^2$.  So that,
	$\pi(\mutil,\betatil,\delta^2|\ytil)$ is proper.
\end{proof}	

The standard MCMC logistic regression exact method is complicated to work with and it takes longer time to get posterior samples. We apply Metropolis Hastings sampler to draw samples for parameters $\betatil$, $\delta^2$ and $\mutil$.

The idea of exact method is to get full conditional posterior distributions for all of the parameters in the model, and then get a large number of independent samples of each parameter with its full conditional posterior density.\\
First, we integrate $\mutil$ from the posterior density to get the joint posterior density of $\betatil, \delta^2|\ytil$ as
$$
\pi(\betatil,\delta^2|\ytil)
\propto\int_{\Omega}^{}\prod_{i=1}^\ell\left\{\prod_{j=1}^{n_i}\frac{e^{(\xtil_{ij}'\betatil_{(0)}+\mu_i){y_{ij}}}}{1+e^{\xtil_{ij}'\betatil_{(0)}+\mu_i}}
\frac{1}{\sqrt{2\pi\delta^2}}e^{-\frac{(\mu_i-\beta_0)^2}{2\delta^2}}\right\}
\frac{1}{(1+\delta^2)^2}d\mutil
$$
$$
=\frac{1}{(1+\delta^2)^2}
\prod_{i=1}^\ell\left\{\int_{-\infty}^{\infty}\frac{e^{\sum\limits_{j=1}^{n_i}(\xtil_{ij}'\betatil_{(0)}+\mu_i){y_{ij}}}}{\prod\limits_{j=1}^{n_i}\left[1+e^{\xtil_{ij}'\betatil_{(0)}+\mu_i}\right]}
\frac{1}{\sqrt{2\pi\delta^2}}e^{-\frac{(\mu_i-\beta_0)^2}{2\delta^2}}d\mu_i
\right\}.
$$
Notice that this is not a simple distribution function for the integration, so we apply numerical integration. Divide the integration domain to $m$ equal intervals $[t_{k-1},t_k], k=1,...,m$. Let $z_i=\frac{\mu_i-\beta_0}{\delta}$ with standard normal distribution. We get an approximate density (very accurate though),
$$
\pi(\betatil,\delta^2|\ytil)\propto  \frac{1}{(1+\delta^2)^2}\left(\frac{1}{\sqrt{\delta^2}}\right)^\ell
\prod_{i=1}^\ell\left\{\sum_{k=1}^{m}\int_{t_{k-1}}^{t_k}\frac{e^{\sum\limits_{j=1}^{n_i}(\xtil_{ij}'\betatil_{(0)}+\mu_i){y_{ij}}}}{\prod\limits_{j=1}^{n_i}\left[1+e^{\xtil_{ij}'\betatil_{(0)}+\mu_i}\right]}
\frac{1}{\sqrt{2\pi}}e^{-\frac{(\mu_i-\beta_0)^2}{2\delta^2}}d\mu_i
\right\}
$$
$$
= \frac{1}{(1+\delta^2)^2}	\prod_{i=1}^\ell\left\{\sum_{k=1}^{m}\int_{t_{k-1}}^{t_k}\frac{e^{\sum\limits_{j=1}^{n_i}(\xtil_{ij}'\betatil_{(0)}+ \beta_0 +z_i\delta){y_{ij}}}}{\prod\limits_{j=1}^{n_i}\left[1+e^{\xtil_{ij}'\betatil_{(0)}+\beta_0+z_i\delta}\right]}
\frac{1}{\sqrt{2\pi}}e^{-\frac{z_i^2}{2}}dz_i
\right\}.
$$

Take the middle point of each interval $[t_{k-1},t_k]$ to estimate the cumulative density function, and denote $\hat{z}_k=\frac{t_k+t_{k-1}}{2}$. We have the following deduction
$$
\pi(\betatil,\delta^2|\ytil)\approx \frac{1}{(1+\delta^2)^2}
\prod_{i=1}^\ell\left\{\sum_{k=1}^{m}\frac{e^{\sum\limits_{j=1}^{n_i}(\xtil_{ij}'\betatil_{(0)}+\beta_0+\hat{z}_k\delta){y_{ij}}}}{\prod\limits_{j=1}^{n_i}\left[1+e^{\xtil_{ij}'\betatil_{(0)}+\beta_0+\hat{z}_k\delta}\right]}
\int_{t_{k-1}}^{t_k}\frac{1}{\sqrt{2\pi}}e^{-\frac{z^2}{2}}dz\right\}.
$$
The integration is now over a standard normal distribution. We consider the interval (-3, 3) for numerical integration, since this domain (standard normal) covers 99.74\% of the distribution that we are dealing with. We take m=100 grid points. Then the joint posterior density for $\betatil$ and $\delta^2$ can be expressed as
\begin{equation}
\pi(\betatil,\delta^2|\ytil)\approx\frac{1}{(1+\delta^2)^2}
\prod_{i=1}^\ell\left\{\sum_{k=1}^{m}\frac{e^{\sum\limits_{j=1}^{n_i}(\xtil_{ij}'\betatil_{(0)}+\beta_0+\hat{z}_k\delta){y_{ij}}}}{\prod\limits_{j=1}^{n_i}\left[1+e^{\xtil_{ij}'\betatil_{(0)}+\beta_0+\hat{z}_k\delta}\right]}
\left(\Phi(t_k)-\Phi(t_{k-1})\right)\right\}.
\end{equation}

We use the Metropolis sampler to draw samples from the joint posterior density of $\pi(\mathunderaccent\tilde{\beta},
\delta^2 \mid \ytil)$ drawing $\mathunderaccent\tilde{\beta}$ and $\delta^2$ simultaneously.	We use
the transformation $\mbox{log}(\delta^2) = \beta_{p+1}$. We obtain a proposal density for $\mathunderaccent\tilde{\beta}$
using a multivariate Student's $t$ density as follows. First, when we run the approximate method, we
obtain the approximate mean, $\mathunderaccent\tilde{\hat{\beta}}$ and covariance matrix $\hat{\Sigma}$
of $\mathunderaccent\tilde{\beta}$. So we take $\mathunderaccent\tilde{\beta} \sim \mbox{Normal}(\mathunderaccent\tilde{\hat{\beta}}, \sigma^2 \hat{\Sigma}), \eta/\sigma^2 \sim \mbox{Gamma}(\eta/2, 1/2)$.
Tuning of the Metropolis sampler is obtained by varying $\eta$ (e.g., $\eta=8$ corresponds to approximately a logistic
random vector). We run the Metropolis sampler in a nonstandard manner, as the Markov chain runs, we reserve
those iterates when the algorithm moves. In this way while tuning, is required to get a reasonable jumping rate,
no other diagnostics (autocorrelation, effective sample size, test of stationarity, thinning) are needed.

The number of nonsampled households in the sampled PSUs are known, but the number of members of a nonsampled household
and their covariates are unknown. For each sampled PSU, we obtain the empirical distribution of the number
of members and their covariates. Then, we use two Bayesian bootstraps to sample the number of members and their
covariates from the pool of the PSU. For both bootstraps we use $1000$ samples.

\section{Quasi-Modes for Logistic Regression}
\def\theequation{B.\arabic{equation}}
\setcounter{equation}{0}


Now we have to specify $\betatil^*_{(0)}$, $\mutil^*$, $\gtil$ and $H$. Consider the log likelihood function
$$
f(\tautil)
=\log h(\tautil) =\log \left\{\prod _{i=1}^\ell\prod _{j=1}^{n_i}\frac{e^{(\xtil_{ij}'\betatil_{(0)}+\mu_i){y_{ij}}}}{1+e^{\xtil_{ij}'\betatil_{(0)}+\mu_i}} \right\}
$$
\begin{equation} 
\label{eq:5}
=\sum_{i=1}^\ell\sum_{j=1}^{n_i} \left\{(\xtil_{ij}'\betatil_{(0)}+\mu_i)y_{ij}-
\log (1+e^{\xtil_{ij}'\betatil_{(0)}+\mu_i} ) \right\}.
\end{equation}

As an estimate of $\mu_i$, we use the empirical logistic transform $z_i$,
$${\hat{\mu}_i}^*=z_i=\log \left\{\frac{\bar{y}_i+\frac{1}{2n_i}}{1-\bar{y}_i+\frac{1}{2n_i}} \right\}.$$
See Appendix C.
Plug ${\hat{\mu}_i}^*$ in the log likelihood function (\ref{eq:5}) and consider it as a function of $\betatil_{(0)}$ only, 
$$ 
g(\betatil_{(0)}) 
=\sum_{i=1}^\ell\sum_{j=1}^{n_i} \left\{(\xtil_{ij}'\betatil_{(0)}+\hat{\mu}_i^*)y_{ij}-
\log (1+e^{\xtil_{ij}'\betatil_{(0)}+\hat{\mu}_i^*} ) \right\}.
$$
The first derivative function of $g(\betatil_{(0)})$ over $\betatil_{(0)}$ is
$$
g'(\betatil_{(0)}) =\sum_{i=1}^\ell\sum_{j=1}^{n_i} \left\{\xtil_{ij} y_{ij}-\frac{\xtil_{ij} e^{(\xtil_{ij}'\betatil_{(0)}+\hat{\mu}_i^*)}}{1+e^{\xtil_{ij}'\betatil_{(0)}+\hat{\mu}_i^*}} \right\}
$$
\begin{equation} 
\label{eq:6}
=\sum_{i=1}^\ell\sum_{j=1}^{n_i} \left\{\xtil_{ij} y_{ij}-\xtil_{ij} (1+e^{-(\xtil'\betatil_{(0)}+\hat{\mu}_i^*)} )^{-1} \right\}.
\end{equation}
Typically, we can solve the equation $g'(\betatil_{(0)})=0$ for the mode as the maximum likelihood estimator (MLE), but here it is not easy to solve the equation because $g'(\betatil_{(0)})$ is complex. We use first order Taylor series approximation to simplify the above function. Since the first order Taylor expansion of $ (1+e^{-(\xtil_{ij}'\betatil_{(0)}+\hat{\mu}_i^*)} )^{-1}$ equals $ (1-e^{-(\xtil_{ij}'\betatil_{(0)}+\hat{\mu}_i^*)} )$, (\ref{eq:6}) equals to
\begin{equation} 
\label{eq:7}
\sum_{i=1}^\ell\sum_{j=1}^{n_i}  \left\{\xtil_{ij} y_{ij}-\xtil_{ij} (1-e^{-(\xtil'_{ij}\betatil_{(0)}+\hat{\mu}_i^*)} ) \right\}.
\end{equation}
This is still complex. We apply Taylor series again and get expansion of the term $e^{-(\xtil'\betatil_{(0)}+\hat{\mu}_i^*)}$ to the first order as $(1-(\xtil'\betatil_{(0)}+\hat{\mu}_i^*))$. Thus (\ref{eq:7}) approximately equals
$$
\sum_{i=1}^\ell\sum_{j=1}^{n_i}  \left\{\xtil_{ij} y_{ij}-\xtil_{ij} (1-(1-(\xtil_{ij}'\betatil_{(0)}+\hat{\mu}_i^*)) ) \right\}
$$
\begin{equation} 
\label{eq:8}
=\sum_{i=1}^\ell\sum_{j=1}^{n_i} \left\{\xtil_{ij}(y_{ij}-\hat{\mu}_i^*)-\xtil_{ij}(\xtil_{ij}'\betatil_{(0)}) \right\}.
\end{equation}
(\ref{eq:8}) is easy to solve. Solve for $g'(\betatil_{(0)})=0$, and we can get the approximate posterior mode of $\betatil_{(0)}$
\begin{equation} 
\label{eq:9}
\betatil_{(0)}^*= \left[ \sum_{i=1}^\ell\sum_{j=1}^{n_i}\xtil_{ij}\xtil_{ij}' \right] ^{-1} \left[ \sum_{i=1}^\ell\sum_{j=1}^{n_i}\xtil_{ij}(y_{ij}-\hat{\mu}_i^*) \right] .
\end{equation}

Plug $\betatil_{(0)}^*$ in the likelihood function (\ref{eq:5}) and consider it as a function of $\mutil$ only, 
$$ 
q(\mu_i)=\log\prod_{j=1}^{n_i}\frac{e^{(\xtil_{ij}'\betatil_{(0)}^*+\mu_i)
		{y_{ij}}}}{1+e^{\xtil_{ij}'\betatil_{(0)}^*+\mu_i}}=\sum_{j=1}^{n_i} \left\{(\xtil_{ij}'\betatil_{(0)}^*+\mu_i)y_{ij}-
\log (1+e^{\xtil_{ij}'\betatil_{(0)}^*+\mu_i} ) \right\}.
$$
The first derivative function of $q(\mu_i)$ over $\mu_i$ is
\begin{equation} 
\label{eq:10}
q'(\mu_i)=\sum_{j=1}^{n_i} \left\{y_{ij}-\frac{e^{(\xtil_{ij}'\betatil_{(0)}^*+\mu_i)}}{1+e^{\xtil_{ij}'\betatil_{(0)}^*+\mu_i}} \right\}=\sum_{j=1}^{n_i} \left\{y_{ij}- (1+e^{-(\xtil_{ij}'\betatil_{(0)}^*+\mu_i)} )^{-1} \right\}.
\end{equation}
Similar to above, we apply Taylor series approximation
$$ (1+e^{-(\xtil_{ij}'\betatil_{(0)}^*+\mu_i)} )^{-1}\approx  (1-e^{-\mu_i}e^{-\xtil_{ij}'\betatil_{(0)}^*} ).$$
So (\ref{eq:10}) equals
$$ 
\sum_{j=1}^{n_i} \left\{y_{ij}- (1-e^{-\mu_i}e^{-\xtil_{ij}'\betatil_{(0)}^*} ) \right\}.
$$
Solve for $q'({\mu_i})=0 $, then the approximate posterior mode (quasi-mode) of $\mu_i$ can be obtained as $${\mu_i}^*=\log \left[ \frac{\sum _{j=1}^{n_i}e^{-\xtil_{ij}'\betatil_{(0)}^*}}{n_i(1-\bar{y}_i)} \right] .$$
Notice that the term $(1-\bar{y}_i)$ in denominator may cause trouble for this posterior mode, because the binary response variable could lead to $\bar{y}_i=1$ for some $i$, so that $(1-\bar{y}_i)=0$. We borrow the idea from the empirical logistic transform (ELT) and make a little adjustment to avoid 0's in denominator
\begin{equation} 
\label{eq:11}
{\mu_i}^*
\approx\log \left[ \frac{\frac{1}{n_i}\sum _{j=1}^{n_i}e^{-\xtil_{ij}'\betatil_{(0)}^*}}{(1-\bar{y}_i+\frac{1}{2n_i})} \right] .
\end{equation}


\section{Empirical Logistic Transform (ELT)}
\def\theequation{C.\arabic{equation}}
\setcounter{equation}{0}

We consider the empirical logistic transform  (ELT) without covariates for binary
data. See Cox and Snell (1972) for the empirical logistic transform (ELT) that
accommodates binary data. 
Letting $y$ denote a binomial random variable with success probability $p$, the empirical logistic transform, $Z$, is
$$Z=\mbox{log} (\frac{Y+\frac{1}{2}}{n-Y+\frac{1}{2}} ),$$
and the corresponding variance V is
$$\mbox{V}=\frac{(n+1)(n+2)}{n(Y+1)(n-Y+1)}.$$
Then, using the approximation of Cox and Snell (1972),
$Z$ has a normal distribution with mean $\theta$ and variance V, where $\theta$ is unknown according to Cox and Snell (1972). 
Suppose $\ytil$ is the variable of length $\ell$. Each of the binary response $y_i(i=1,...,\ell)$ follows a binomial distribution with corresponding number of observations $n_i$ and probability $p_i$. The goal is to estimate the Bernoulli probability parameter $p_i$. Here we assume that
$$
y_i\stackrel{ind}{\sim} \mbox{Binomial}\left\{n_i,p_i\right\}
$$
and for logistic transform we
define $z_i=\log(\frac{y_i+\frac{1}{2}}{n_i-y_i+\frac{1}{2}})$ as the the empirical logistic transforms, and	$V_i=\frac{(n_i+1)(n_i+2)}{n_i(y_i+1)(n_i-y_i+1)}$ as the associated variances. Then,
$$
z_i \mid \mu_i \stackrel{ind} \sim \mbox{Normal}(\mu_i, V_i).
$$
We can actually start with this approximation based on the empirical logistic transform. However, this approximation
will not work for binary data with covariates at the unit level, but we will make use it in our approximation for
logistic regression with binary data in a less important manner.

\end{appendices}


\noindent
{\bf Acknowledgments}\\
This work is a paper to honor JNK Rao on his  eightieth birthday.
This research was supported by a grant from the Simons Foundation (\#353953, Balgobin Nandram).


\begin{mybibliography}{}
		
\bibitem[]{} Albert, J. H. and Chib, S. (1993). 
\newblock  {Bayesian Analysis of Binary and Polychotomous Response Data}.
\newblock  {\it Journal of the American Statistical Association}, \textbf{88}, 669-679.		
		
\bibitem[]{} Beaumont, M. A., Zhang, W. and Balding, D. J. (2002). 
\newblock  {Approximate Bayesian Computation in Population Genetics}.
\newblock  {\it Genetics}, \textbf{162}, 2025-2035.		
				
\bibitem[]{} Boyd, S. and Vandenberghe, L. (2004). 
\newblock  {\it Convex Optimization}.
\newblock  Cambridge University Press, Cambridge, United Kingdom.

\bibitem[]{} Central Bureau of Statistics Thapathali. 
\newblock {Nepal living Standards Survey 2003/04}.
\newblock  {\it Statistical Report Volume 1}, Kathmandu, Nepal.

\bibitem[]{} Chen, M-H, Ibrahim, J. G. and Kim, S. (2008).
\newblock {Properties and Implementation of Jeffreys's Prior
in Binomial Regression Models}.
\newblock {\it Journal of the American Statistical Association}, \textbf{103}, 1659-1664.

\bibitem[]{} Cox, D. R. and Snell, E. J. (1972). 
\newblock {\it Analysis of Binary Data, Second Edition}.
\newblock  {Chapman \ Hall}. \textbf{2}, 31-32.
		
\bibitem[]{} Dharmadhikari, S. W. and Joag-Dev, K. (1988). 
\newblock    {\it Unimodality, Convexity and Applications}.
\newblock    Academic Press, New York.
			
\bibitem[]{} Fearnhead, P. and Prangle, D. (2012).
\newblock    { Constructing Summary Statistics for Approximate
               Bayesian computation: Semi-automatic ABC}. 
\newblock    {\it Journal of the Royal Statistical Society},  \textbf{B(74)}, 419-474. 

\bibitem[]{}  Faes, C., Ormerod, J. T. and Wand, M. P. (2011).
\newblock {Variational Bayesian Inference for Parametric and
Nonparametric Regression With Missing Data}.
\newblock  {\it Journal of the American Statistical Association}, \textbf{106}, 959-971.

\bibitem[]{} Fay, R.E. and Herriot, R.A. (1979).
\newblock {Estimates of Income for Small Places: an Application of James-Stein Procedures to Census Data}.
\newblock  {\it Journal of the American Statistical Association}, \textbf{74}, 269-277.

\bibitem[]{} Ferkingstad, E. and Rue, H. (2015).
\newblock {Improving the INLA approach for approximate Bayesian inference for latent Gaussian models}.
\newblock  {\it Electronic Journal of Statistics}, \textbf{9}, 2706-2731.
		
\bibitem[]{} Fong, Y., Rue, H. and Wakefield, J. (2010). 
\newblock {Bayesian inference for generalized linear mixed models}.
\newblock  {\it Biostatistics}, \textbf{11(3)}, 397-412.

\bibitem[]{} Ghosh, M., Natarajan, K., Stroud, T. W. F. and Bradley P. Carlin, B P. (1998).
\newblock {Generalized Linear Models for Small-Area Estimation}.
{\it Journal of the American Statistical Association}, \textbf{93, 441}, 273-282   

\bibitem[]{} Held, L., Schr\"{o}dle, B. and Rue, H. (2010). 
\newblock {Posterior and Cross-Validatory Predictive Checks: A Comparison of MCMC and INLA}.
\newblock  {\it Statistical Modeling and Regression Structures, Physica-Verig HD}, 91-110.

\bibitem[]{} Holmes, C. C. and Held, L. (2006). 
\newblock {Bayesian Auxiliary Variable Models for Binary and Multinomial Regression}.
\newblock  {\it Bayesian Analysis}, \textbf{1(1)}, 145-168.
		
\bibitem[]{} Illian, J. B., S\o rbye, S. H. and Rue, H. (2012). 
\newblock {A toolbox for fitting complex spatial point process models using integrated nested 
		            Laplace approximation (INLA)}.
\newblock  {\it The Annals of Applied Statistics}, \textbf{6(4)}, 1499-1530.
		

\bibitem[]{} Jaakkola, T. S. and Jordan, M. I. (2000).
\newblock    {Bayesian Parameter Estimation via Variation Methods}.
\newblock    {\it Statistics and Computing}, \textbf{10}, 25-37.
		
\bibitem[] {} Liu, J. and Dey, D. K. (2008).
\newblock {Skew random effects in multilevel binomial models:
an alternative to nonparametric approach}.
\newblock {\it Statistical Modelling}, \textbf{8(3)}, 221-241

\bibitem[] {} Larsen, K., Petersen, J. H., Budtz-Jørgensen, E. and 
            Endahl, L. (2000).
\newblock {Interpreting Parameters in the Logistic Regression Model with Random Effects}.
\newblock {\it Biometrics},  \textbf{56(3)}, 909-914.		

\bibitem[]{} Miroshnikov, A. and Conlon, E. M. (2015).
\newblock  {Parallel Markov Chain Monte Carlo for Non-Gaussian Posterior Distributions}.
\newblock   {\it Department of Mathematics and Statistics, University of Massachusetts, Amherst,
             Massachusetts}, pp. 1-32.

\bibitem[]{} Nandram, B. (2000).
\newblock   {Bayesian Generalized Linear Models for Inference about Small Areas}.
\newblock   {\it In Generalized linear models: A Bayesian Perspective},
            Eds. D. K. Dey, S.K. Ghosh and B.K. Mallick, Marcel
            Dekker, Chapter 6, 91-114.

\bibitem[]{}  Nandram, B. (1989).
\newblock  {Discrimination between Complementary
             Log-log and Logistic Model for Ordinal Data}.
\newblock   {\it Communications in Statistics, Theory and Methods}, \textbf{18}, 2155-2164.

\bibitem[]{} Nandram, B. and Chen, M-H. (1996).
\newblock   {Reparameterizing the Generalized Linear Model to Accelerate Gibbs Sampler Convergence}.
\newblock   {\it Journal of Statistical Computation and simulation}, \textbf{54}, 129-144.

\bibitem[]{} Nandram, B. and Choi, J. W. (2010).
\newblock   {A Bayesian Analysis of Body Mass Index Data from 
            Small Domains Under Nonignorable Nonresponse and Selection}.
\newblock   {\it Journal of the American Statistical Association}, \textbf{105}, 120-135.

\bibitem[]{} Nandram, B. and Erhardt, E. (2005).
\newblock    {Fitting Bayesian Two-Stage Generalized Linear Models
              Using Random Samples via the SIR Algorithm}.
\newblock    {\it Sankhya}, \textbf{66}, 733-755.

\bibitem[]{} Ormerod, J. T., and Wand, M. P. (2010). 
\newblock   {Explaining Variational Approximations}.
\newblock   {\it The American Statistician}, \textbf{64}, 140-153.

\bibitem[]{} Paulino, C.D., Silva, G. and Achcar, J.A. (2005). 
\newblock   {Bayesian Analysis of Correlated Misclassified Binary Data}.
\newblock   {\it Computational Statistics \& Data Analysis}, \textbf{49(4)}, 1120-1131. 

\bibitem[]{} Potthoff, R.F., Woodbury, M.A. and Manton, K.G. (1992). 
\newblock   {Equivalent Sample Size and Equivalent Degrees of Freedom Refinements for Inference Using Survey Weights 
under Superpopulation Models}.
\newblock   {\it Journal of the American Statistical Association}, \textbf{87}, 383-396.

\bibitem[]{} Roberts, G., Rao, J. N. K. and Kumar, S. (1987).
\newblock    {Logistic Regression Analysis of Sample Survey Data}.
\newblock {\it Biometrika}, \textbf{74(1)}, 1-12.
	
\bibitem[]{} Rubin, D. B. (1981).
\newblock    {The Bayesian Bootstrap}.
\newblock {\it The Annals of Statistics}, \textbf{9(1)}, 130-134.
		
\bibitem[] {} Rue, H., Martino, S. and Chopin, IN. (2009). 
\newblock {Approximate Bayesian Inference for Latent Gaussian Models Using Integrated Nested 
		            Laplace Approximations}.
\newblock  {\it Journal of the Royal Statistical Society} Series B, \textbf{71(2)}, 319-392.

\bibitem[] {} Santos, C. C., Loschi, R. H. and Arellano-Valle, R. B. (2013).
\newblock {Parameter Interpretation in Skewed Logistic Regression 
            with Random Intercept}.
\newblock {\it Bayesian Analysis}, \textbf{8(2)}, 381-410

\bibitem[] {} Souza, A.D. and Migon, H.S. (2010).
\newblock {Bayesian outlier analysis in binary regression}.
\newblock {\it Journal of Applied Statistics}, \textbf{37(8)}, 1355-1368

\bibitem[] {} Scott, S. L., Blocker, A. W., Bonassi1, F. V., Chipman, H. A., George, E. I. and McCulloch, R. E. (2013).
\newblock {Bayes and Big Data: The Consensus Monte Carlo Algorithm}.
\newblock {\it Technical Report}, Google, Inc., 1-22.

\bibitem[] {} Torabi, M. and Rao, J. N. K. (2014).
\newblock {On Small Area Estimation under a Sub-Area Level Model}.
\newblock {\it Journal of Multivariate Analysis},\textbf{127}, 36-55.

\end{mybibliography}

\captionsetup[table]{skip=0pt} 
\begin{table}[htb]
\caption{\textbf{Categorical tables for 60,221 households by posterior mean of model (2) with random effects at the
ward level projected to the households and the model (1) at household level}}
\label{TAB:PMcat}
\begin{center}
\begin{tabular}{cccccccccccccccccccccccccccccc}
   && & \multicolumn{5} {c} {Model 2} &&  \\
\hline
 Model 1          && &  $.0-.2$ &  $.2-.4$  & $.4-.6$ & $.6-.8$  & $.8-1$    &&  Total\\         
\hline  

    $.0-.2$   &&   &   ~~1   &     ~~19  &     ~~586   &    ~228   &     ~29     &&       ~~863\\
    $.2-.4$   &&   &   ~~2   &    ~187  &    ~1845   &    ~139   &     ~44     &&      ~2217\\
    $.4-.6$   &&   & 251   &   1730  &   29861   &   9638   &    664     &&     42144\\
    $.6-.8$   &&   &  ~87   &    ~548  &    ~9115   &   3491   &    253     &&     13494\\
    $.8-1$   &&   &   ~~9   &     ~~24  &    ~1022   &    ~384   &     ~64     &&      ~1503\\

\hline
\end{tabular}
\\
\flushleft
\end{center}
\vspace*{.30in}
\caption{\textbf{Categorical tables for 60,221 households by posterior standard deviation of model (2) with random effects at the
ward level projected to the households and the model (1) at household level}}
\label{TAB:PScat}
\begin{center}
\begin{tabular}{cccccccccccccccccccccccccccccc}

   && & \multicolumn{5} {c} {Model 2} &&  \\
\hline
 Model 1      && &  $.0-.1$ & $.1-.2$   & $.2-.3$ & $.3-.4$  & $.4-.5$   &&  Total\\         
\hline  

    $.0-.1$   &&  &  ~~5   &   ~~66   &  ~~~626    &   ~~41    &    ~1     &&     ~~739\\
    $.1-.2$   &&  &  ~13   &   ~233   &   ~2411    &   ~105    &    ~2     &&     ~2764\\
    $.2-.3$   &&  &  348   &   2826   &  46500     &   3391    &    37     &&     53102\\
    $.3-.4$   &&  &  ~~5   &   ~180   &  ~2859     &   ~553    &    ~3     &&     ~3600\\
    $.4-.5$   &&  &  ~~0   &   ~~~0   &  ~~~~13    &   ~~~2    &    ~1     &&     ~~~16\\

\hline
\end{tabular}
\\
\flushleft
\end{center}
\vspace*{.30in}

\caption{\textbf{Categorical tables for 60,221 households by posterior coefficient of variation of model (2) with random effects at the
ward level projected to the households and the model (1) at household level}}
\label{TAB:PCVcat}
\begin{center}
\begin{tabular}{cccccccccccccccccccccccccccccc}

 && & \multicolumn{5} {c} {Model 2} &&  \\
\hline
Model 1     && & $.02-.05$ & $.05-.10$ & $.10-.25$  & $.25-.50$  & $>.50$   &&  Total\\         
\hline  
 $.02-.05$      && &    ~0  &     ~~0   &    ~~4   &  ~~~20  &   ~~~~9    &&      ~~~33\\
 $.05-.10$      && &    ~0  &     ~~3   &    ~20   &  ~~188  &   ~~~66    &&      ~~277\\
 $.10-.25$      && &    ~1  &     ~~4   &    ~80   &  ~~856  &   ~~368    &&      ~1309\\
 $.25-.50$      && &    28  &     128   &    791   &  23694  &   10535    &&      35176\\
 $>.50$      && &    22  &     ~62   &    399   &  ~9996  &   12947    &&      23426\\

\hline
\end{tabular}
\\
\flushleft
\end{center}
\end{table}

\newpage
\begin{figure}[htb]
	\begin{center}
		\caption{\textbf{Comparison of the posterior means (PM) of the household proportions
			by the approximate method 
			and the MCMC method}}
		\includegraphics[height=8in,width=6in]{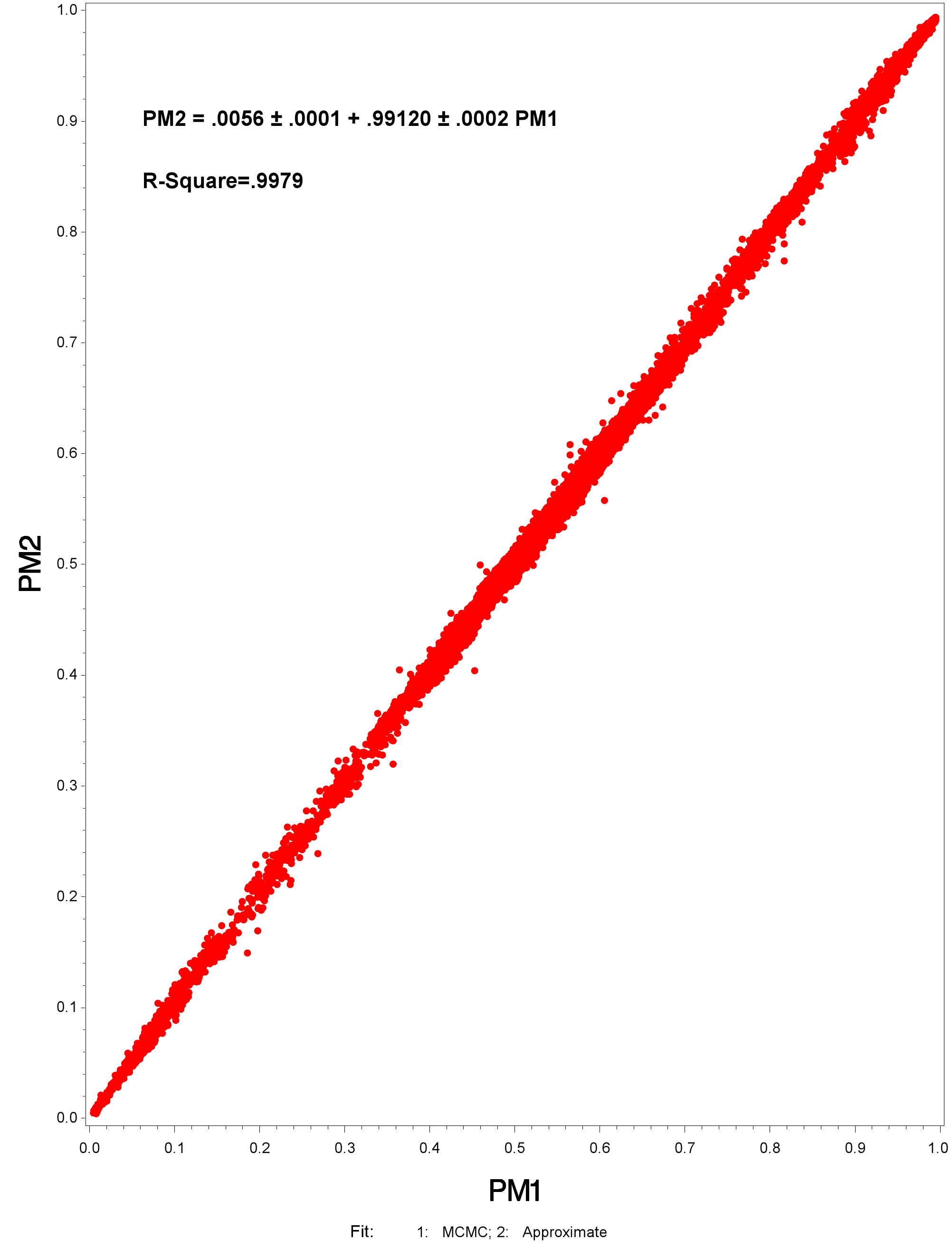}
		\label{FG:poPM}
	\end{center}
\end{figure}

\begin{figure}[htb]
	\begin{center}
		\caption{\textbf{Comparison of the posterior standard deviations (PSD) of the household
			proportions by the approximate method  and the MCMC method}}
		\includegraphics[height=8in,width=6in]{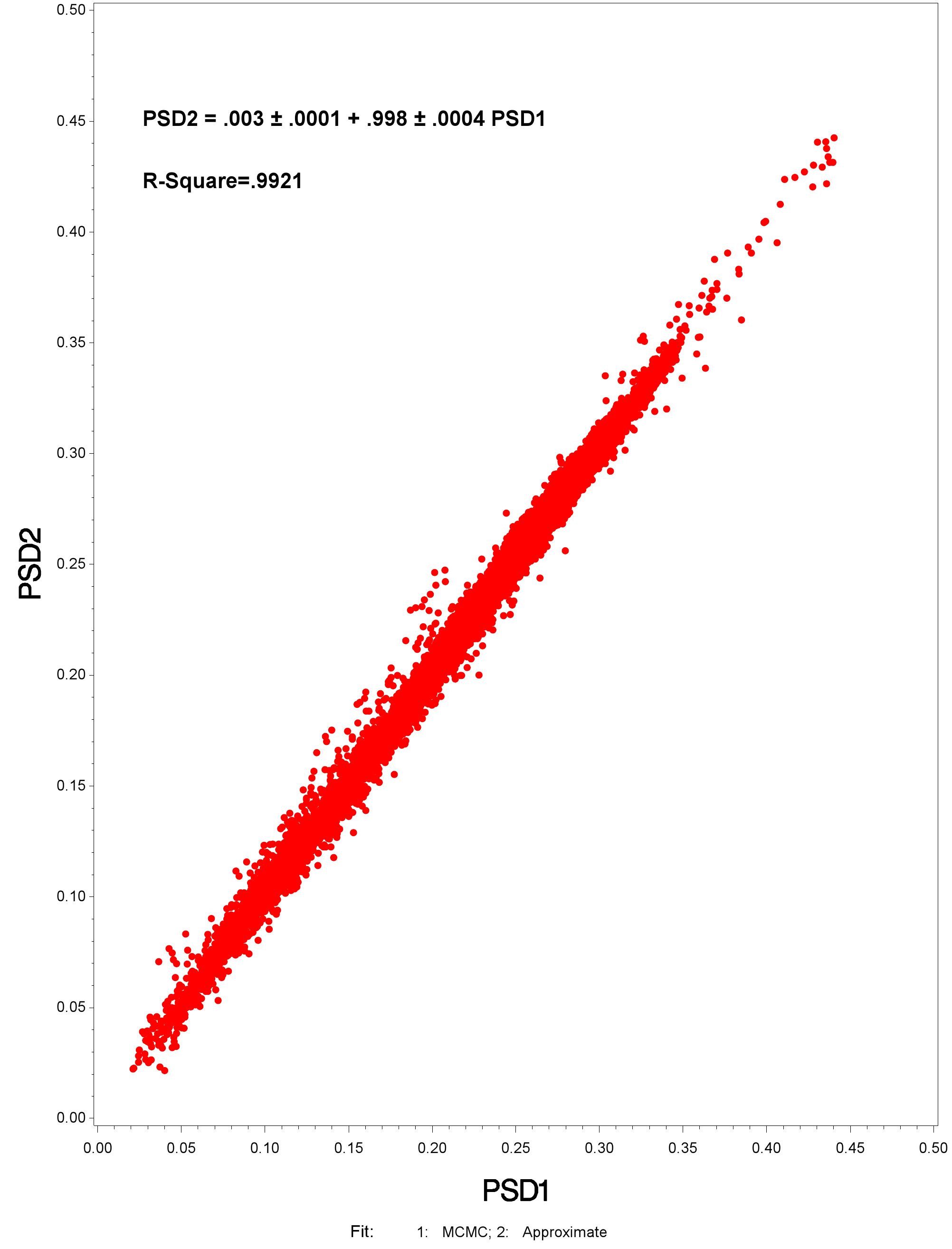}
		\label{FG:poPSD}
	\end{center}
\end{figure}

\begin{figure}[htb]
	\begin{center}
\caption{\textbf{Comparison of the posterior coefficient of variations (CV) of the household proportions by the approximate method  and the MCMC method}}
\includegraphics[height=8in,width=6in]{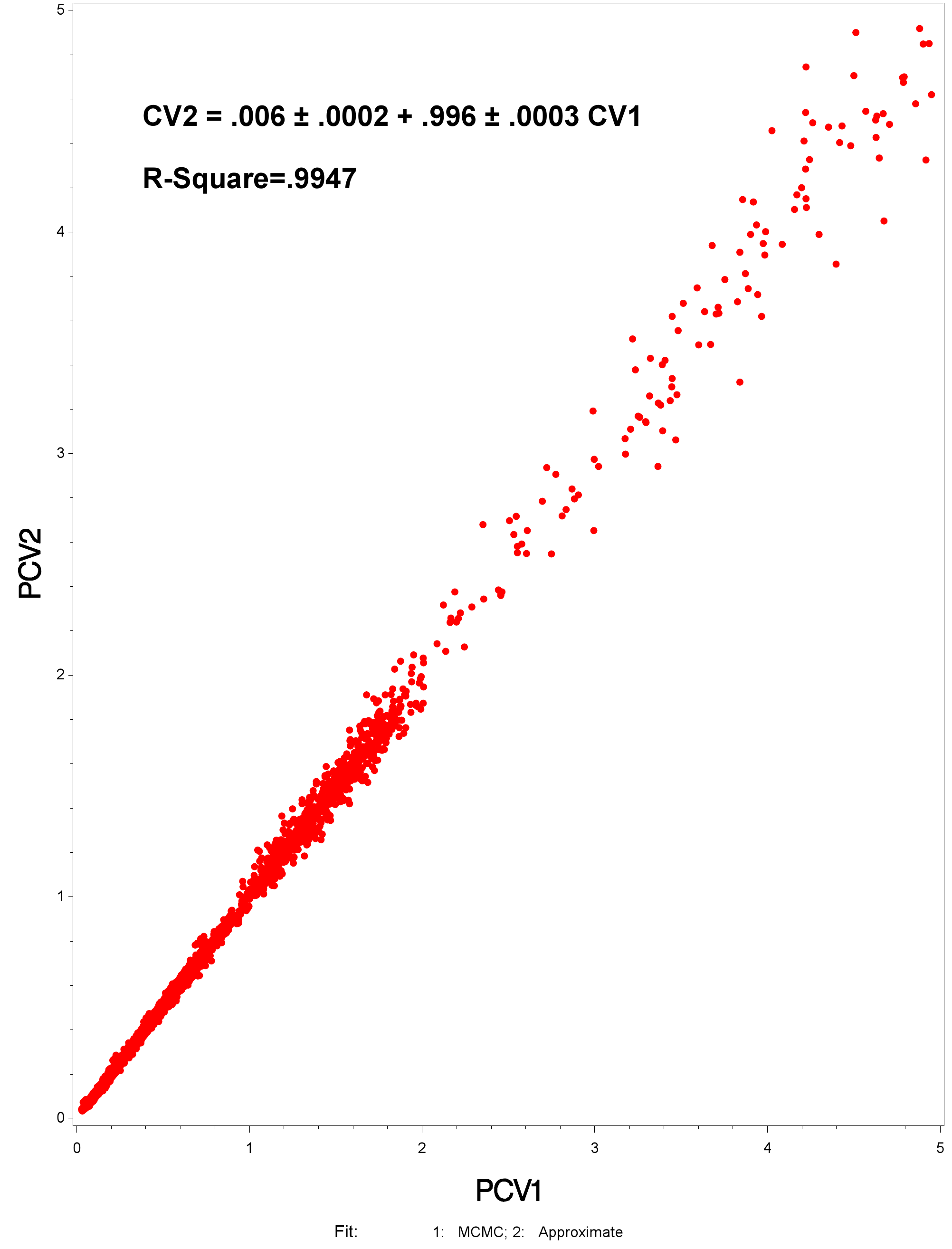}
\label{FG:poCV}
\end{center}
\end{figure}

\begin{figure}[htp]
	\begin{center}
		\caption{\textbf{Comparison of the posterior means (PM) of the ward proportions
					by the approximate method and the MCMC method}}
	\includegraphics[height=8in,width=6in]{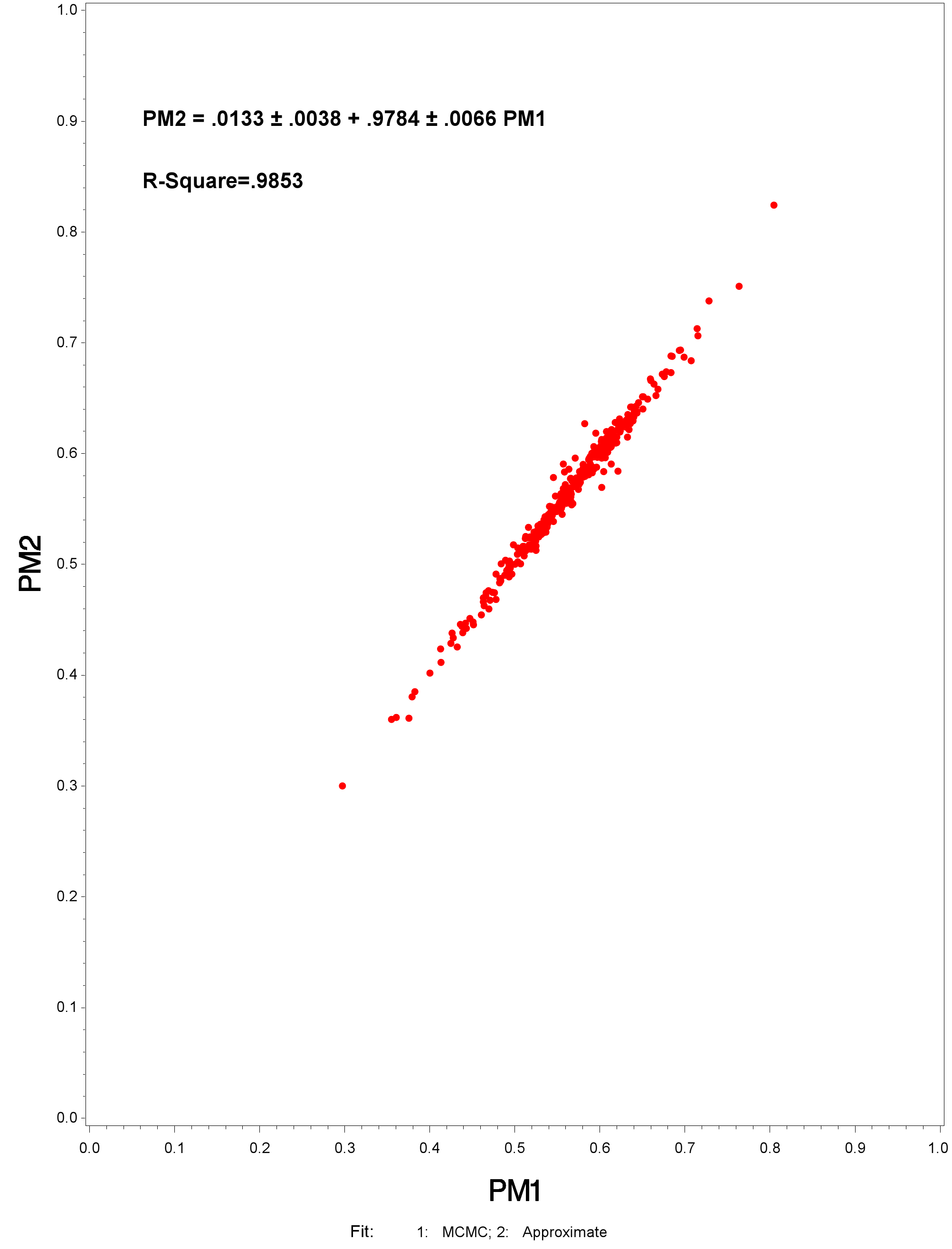}			
\label{FG:poPMp}
	\end{center}
\end{figure}

\begin{figure}[htb]
	\begin{center}
		\caption{\textbf{Comparison of the posterior standard deviations (PSD) of the ward
			proportions by the approximate method  and the MCMC method}}
			\includegraphics[height=8in,width=6in]{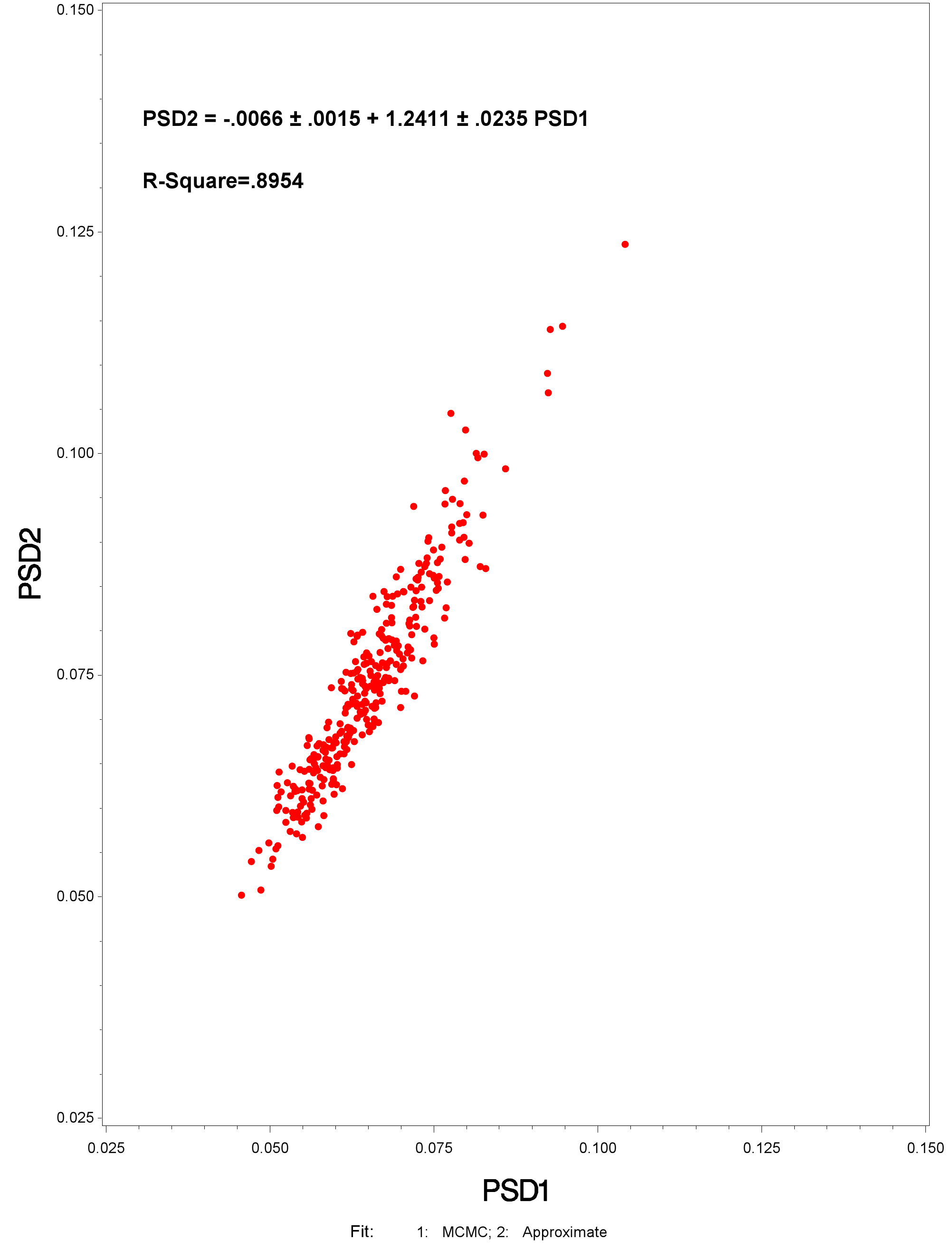}	
		\label{FG:poPSDp}
	\end{center}
\end{figure}

\begin{figure}[htb]
	\begin{center}
		\caption{\textbf{Comparison of the posterior coefficient of variations (CV) of the  ward proportions by the  approximate method and the MCMC method}}
		\includegraphics[height=8in,width=6in]{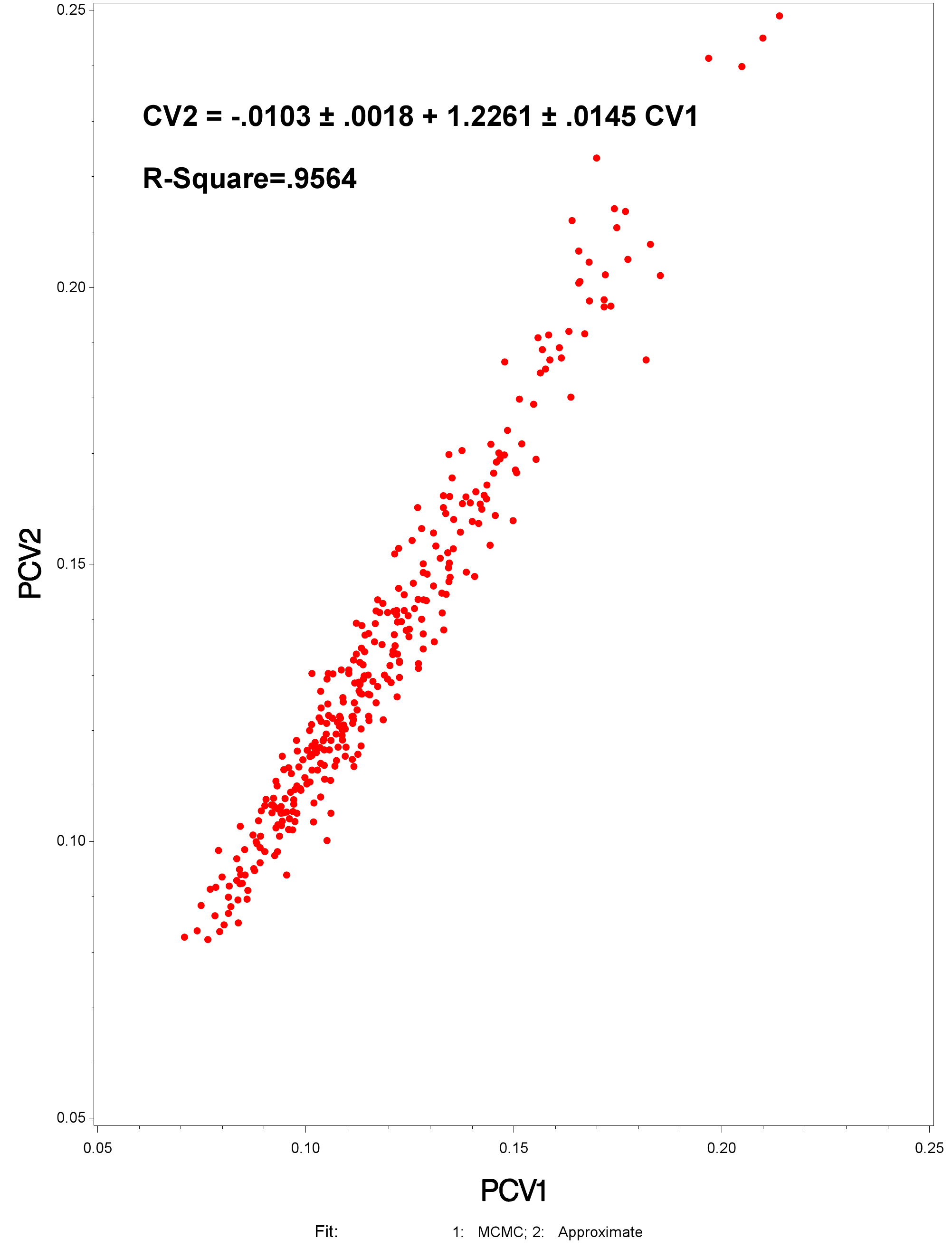}	
		\label{FG:poCVp}
	\end{center}
\end{figure}

\begin{figure}[htb]
	\begin{center}
		\caption{\textbf{Plots of the empirical posterior densities of $\delta^2$ and
			$\beta_0$ for the approximate method and the MCMC method}} 
		\includegraphics[height=8in,width=6in]{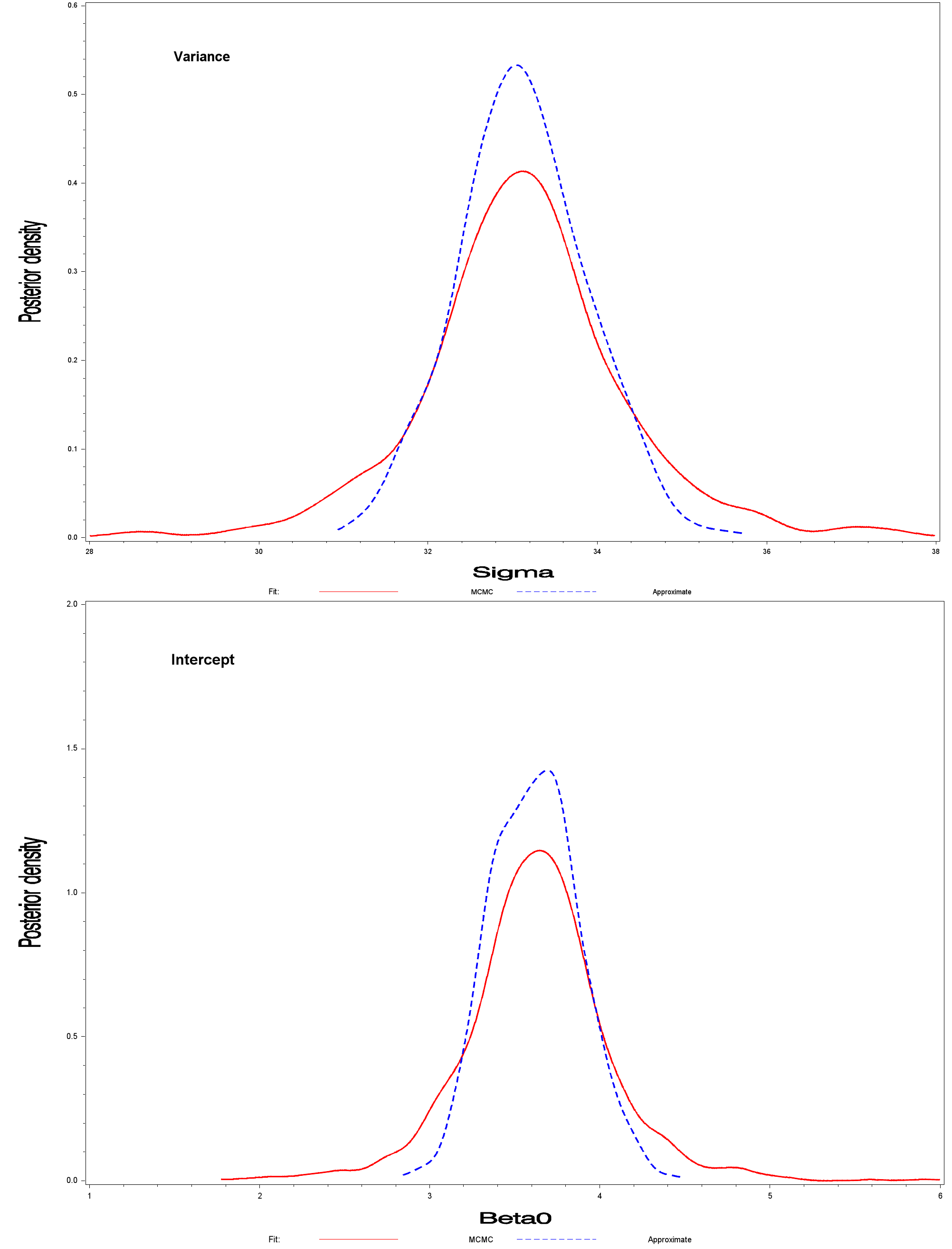}
		\label{FG:podb}
	\end{center}
\end{figure}

\begin{figure}[htb]
	\begin{center}
		\caption{\textbf{Plots of the empirical posterior densities of regression coefficients
			for the approximate method  and the MCMC method}}
		\includegraphics[height=8in,width=6in]{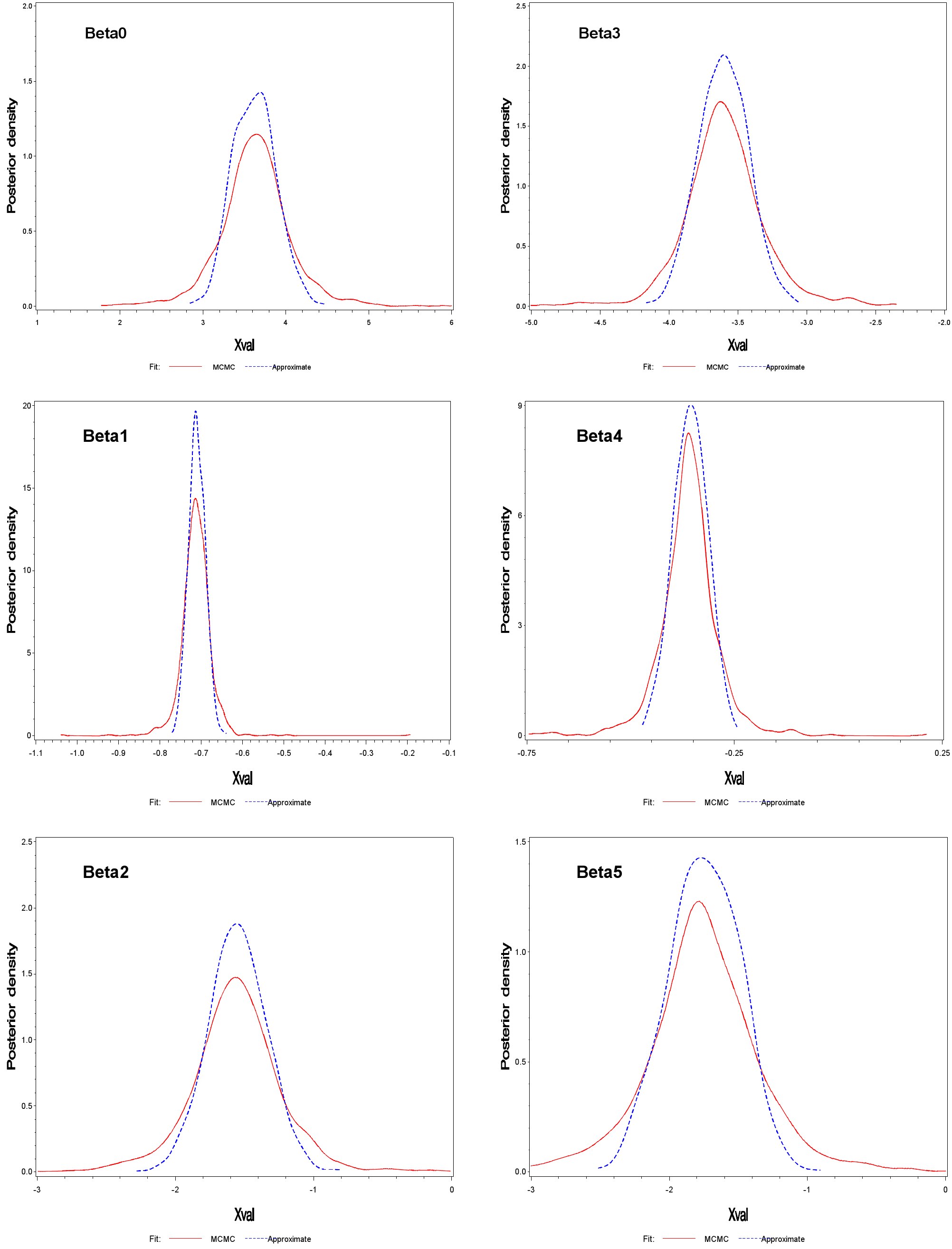}
		\label{FG:pobeta}
	\end{center}
\end{figure}

\begin{figure}[htb]
	\begin{center}
		\caption{\textbf{Comparison of the posterior means (PM) of the household proportions
			by the two MCMC methods}}
			\includegraphics[height=8in,width=6in]{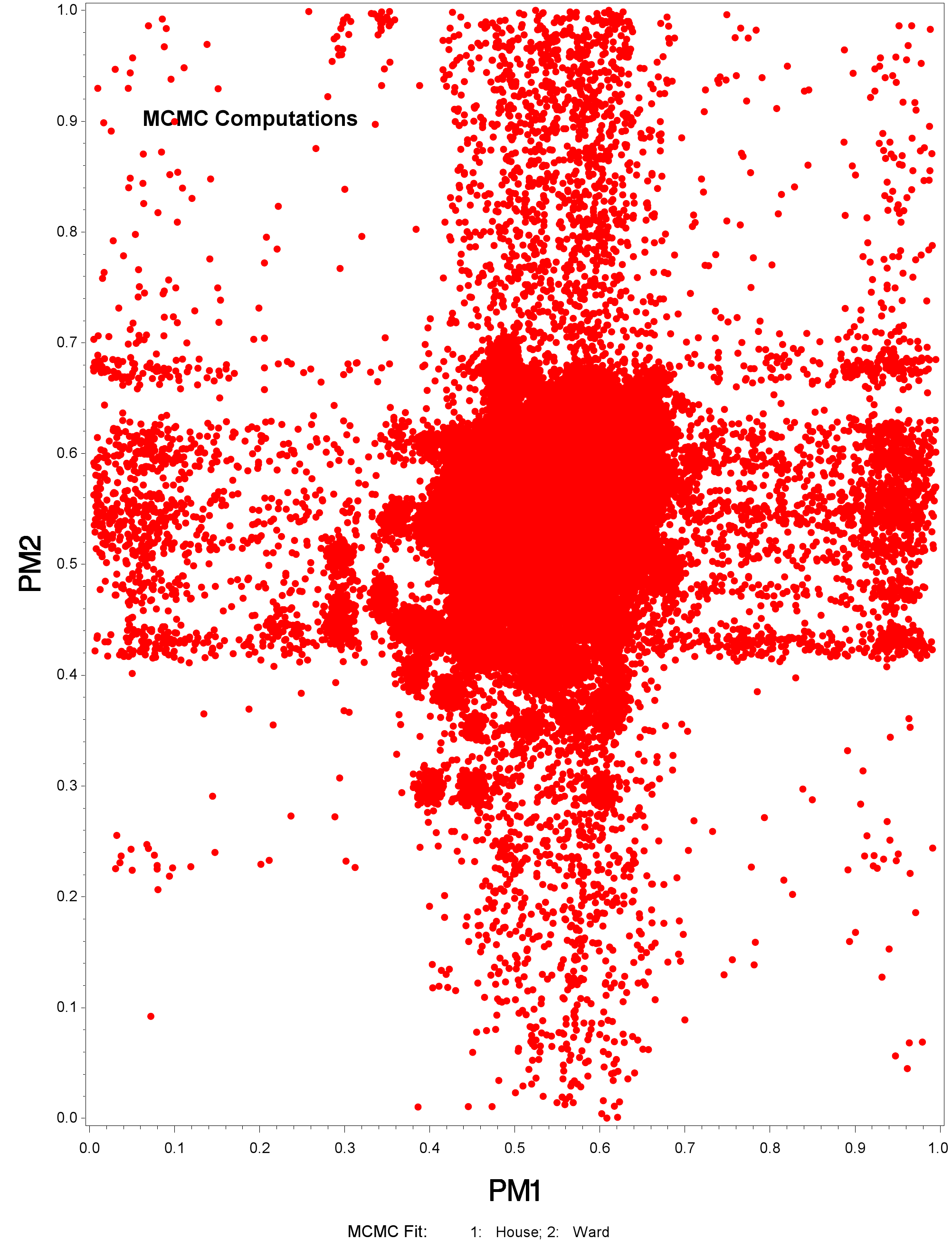}
		\label{FG:poPMe}
	\end{center}
\end{figure}

\begin{figure}[htb]
	\begin{center}
		\caption{\textbf{Comparison of the posterior standard deviations (PSD) of the household proportions by the two MCMC methods}}
		\includegraphics[height=8in,width=6in]{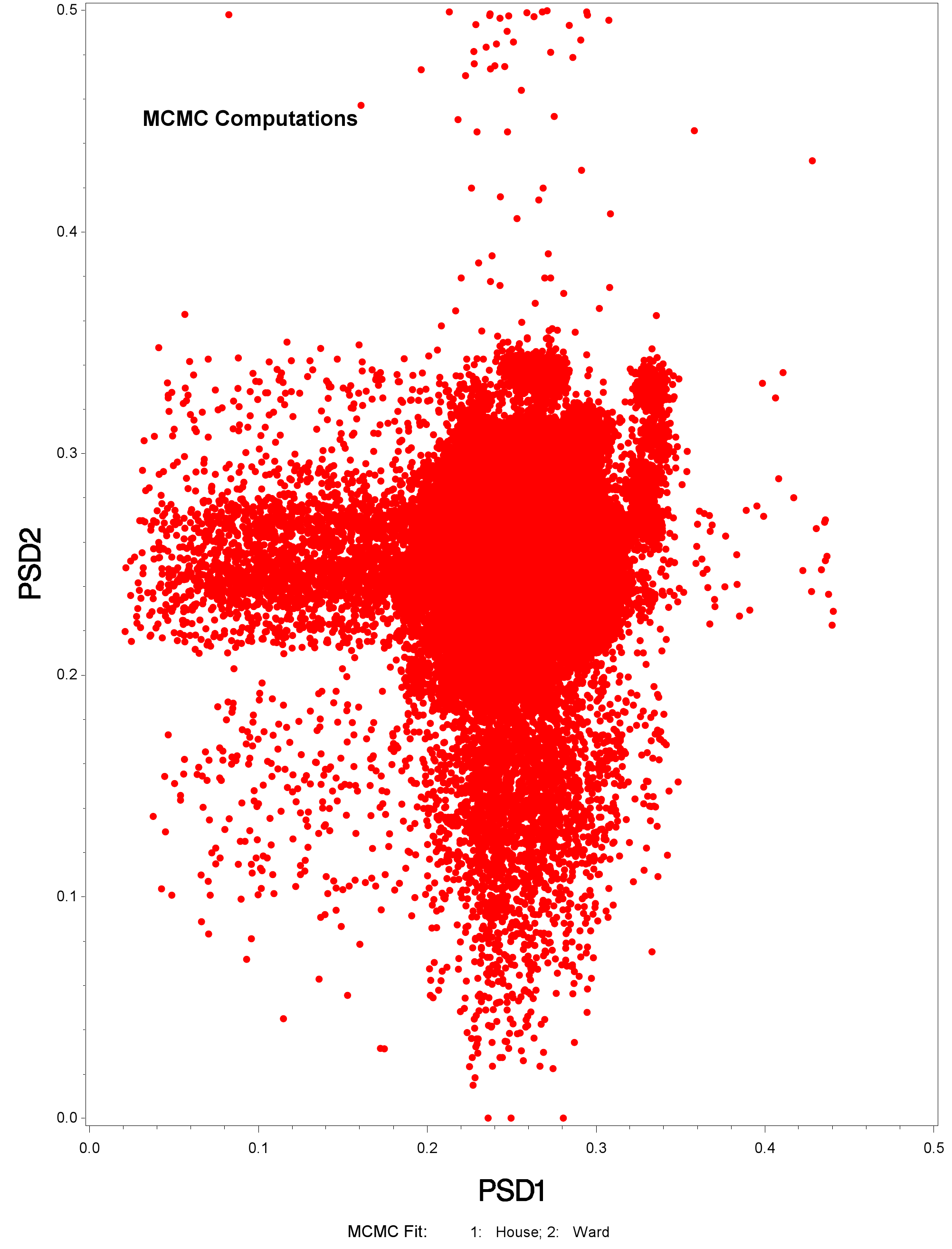}
		\label{FG:poPSDe}
	\end{center}
\end{figure}

\begin{figure}[htb]
	\begin{center}
		\caption{\textbf{Comparison of the posterior coefficient of variations (CV) of the  household proportions by the two MCMC methods}}
		\includegraphics[height=8in,width=6in]{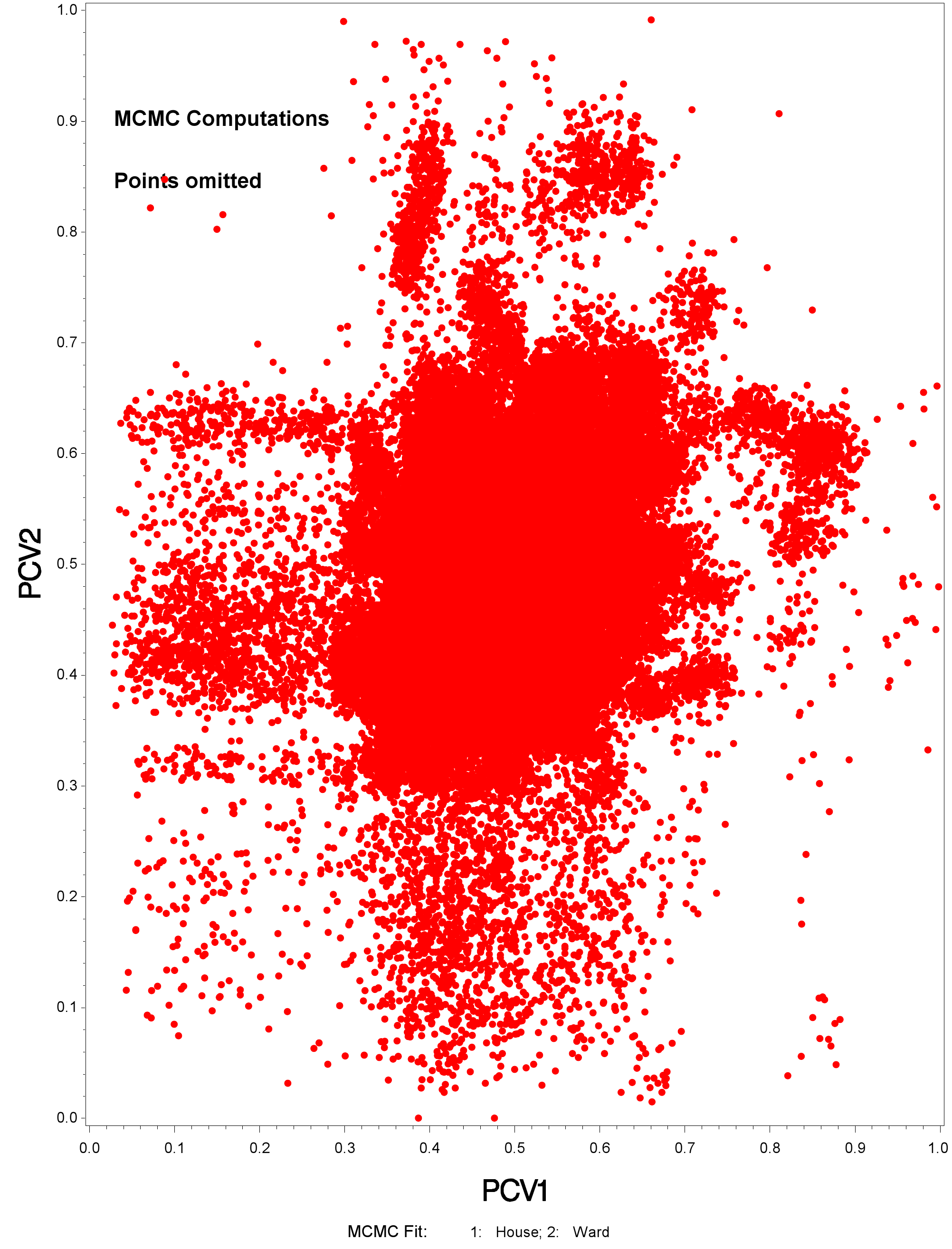}
		\label{FG:poCVe}
	\end{center}
\end{figure}

\end{document}